\documentclass{article}
\usepackage{booktabs} 
\usepackage[english]{babel}
\usepackage{tikz}
\usepackage{amssymb}
\usepackage[T1]{fontenc}
\usepackage[utf8]{inputenc}
\usepackage{calrsfs}
\usepackage{amsmath}
\usepackage{amsthm}
\usepackage{enumerate}
\usepackage{caption}
\usepackage{a4wide}
\usepackage{authblk}
\usepackage{algorithm}
\usepackage{algpseudocode}
\usepackage{algorithmicx}
\usepackage{graphicx}
\usepackage[mathscr]{eucal}
\usepackage[title,titletoc,toc]{appendix}
\usetikzlibrary[shapes.misc]
\usetikzlibrary[shapes.geometric]
\newtheorem{de}{Definition}
\newtheorem{theo}{Theorem}
\newtheorem{prop}{Proposition}

\newtheorem{lem}{Lemma}

\newtheorem{re}{Remark}

\tikzstyle{vertex}=[circle, draw, inner sep=0pt, minimum size=6pt]

\date{\today}

\title{Distributed leader election and computation of local identifiers for programmable matter}
\author[1]{Nicolas Gastineau}
\author[1]{Wahabou Abdou}
\author[1]{Nader Mbarek}
\author[1]{Olivier Togni}
\affil[1]{LE2I FRE2005, CNRS, Arts et Métiers, Université Bourgogne Franche-Comté, F-21000 Dijon, France}
\begin{document}

\maketitle

\begin{abstract}
The context of this paper is programmable matter, which consists of a set of computational elements, called {\em particles}, in an infinite graph. 
The considered infinite graphs are the square, triangular and king grids. Each particle occupies one vertex, can communicate with the adjacent particles, has the same clockwise direction and knows the local positions of neighborhood particles. 
Under these assumptions, we describe a new leader election algorithm affecting a variable to the particles, called the $k$-local identifier, in such a way that particles at close distance have each a different $k$-local identifier. For all the presented algorithms, the particles only need a $O(1)$-memory space.
\end{abstract}

\section{Introduction}

Programmable matter can be seen as modular robots (called modules or particles) able to fix to adjacent modules and send (receive) messages to (from) other modules fixed to the entity. Thus, the different modules form a geometric shape which is a network.
Usually, a module can fix to another module using a finite number of ports (see Figure~\ref{sphere} for an example of spherical modules). Also, the modules know the ports that are in contact with other modules and have a knowledge about the geographic position of their ports. Moreover, the ports are supposed to be homogeneously distributed along the surface of each module. Such assumptions imply that the way how the modules are on a plane can be modeled by a grid. In this paper, we only consider modules on a plane surface, i.e. two dimensional grids.
In this context, the geometric amoebot model~\cite{ZD2014,ZD2015,ZD2015b,ZD2016a,ZD2016b,ZD2017} aims to model the properties of a network for programmable matter.

Distributed algorithms aim to give a theoretical algorithmic framework in order to model the execution of an algorithm that runs on a network of computational elements that can cooperate in order to solve network problems.
In distributed algorithm frameworks, it is often supposed that the different elements of the network do not have a unique identity, i.e., the network is anonymous. In anonymous networks, a natural question is how to perform a leader election, i.e., how to determine a singular element in an anonymous network.
It is well known that for some network structures, the ring for example, there is no deterministic leader election algorithm~\cite{AW2004}.

In 1999, Antoni Mazurkiewicz~\cite{AM1999} has presented a deterministic general algorithm to determine a leader (in the case it is possible to do so). 
In the situation where the elements have access to a random source, then it is also proven that no algorithm can correctly determine a leader in a ring  with any probability $\alpha>0$ \cite{IR1981}. Due to the assumption we make about the ports of the particles in the context of programmable matter (a particle knows the ports which are in contacts with other particles and knows the geographic position of its ports), the leader election problem becomes different than in the classical system. In particular, in the field of programmable matter, there exists a probabilistic algorithm that determine a leader (and in particular for a ring) with probability $1$~\cite{JJD2017}.

Several projects aim to build programmable matter prototypes. One of such projects \cite{BOUR2018,TU2017}, financed by the french National Agency for Research, aims to build cuboctahedral particles able to deform them-selves in order to move.
This project can be split in two phases, one consists in manufacturing the hardware of prototype matters, the second consists in proposing algorithms for programmable matter.
The final goal of this project is to sculpt a shape-memory polymer sheet with programmable matter. 
In the continuity of the algorithm phase of this project \cite{BOUR2018}, we propose algorithms for the self-configuration, i.e., in order to create identifiers and spanning trees.

In the context of programmable matter \cite{BOUR2015,BUT2004,HI2014,NAZ2016,TU2017,YI2007}, it is supposed that a network can contain several millions of modules and that each module has possibly a nano-centimeter size. These two facts lead us to believe that even a $O(\log(n))$-space memory for each module, $n$ being the number of modules, is not technically possible. Also, because of the large number of modules, it can be very challenging and time consuming to implement a unique identity to the modules when they are created. In this context, we suppose that the modules can not store a unique identity, i.e., that the network is anonymous. In this paper we propose deterministic $O(1)$-space memory algorithms to determine a leader in the network and to create $k$-local identifiers of the particles. A $k$-local identifier is a variable affected to each module of the network which is different for every two modules at distance at most $k$.
Note that leader election \cite{JJD2017,DIL2017} plays a significant role in numerous problems of programmable matter.

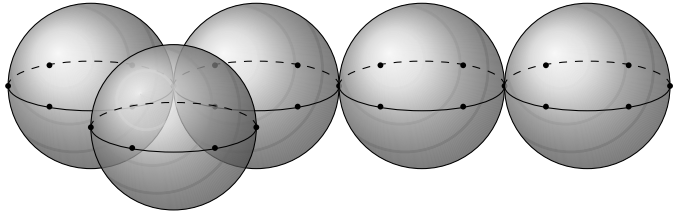
\begin{figure}[t]
\begin{center}
\begin{tikzpicture}[scale=0.55]
  \shade[ball color = gray!40, opacity = 0.4] (0,0) circle (2cm);
  \draw (0,0) circle (2cm);
  \draw (-2,0) arc (180:360:2 and 0.6);
  \draw[dashed] (2,0) arc (0:180:2 and 0.6);
  \fill[fill=black] (2,0) circle (2pt);
  \fill[fill=black] (-2,0) circle (2pt);
  \fill[fill=black] (-1,-0.5) circle (2pt);
  \fill[fill=black] (1,-0.5) circle (2pt);
  \fill[fill=black] (-1,0.5) circle (2pt);
  \fill[fill=black] (1,0.5) circle (2pt);
  \shade[ball color = gray!40, opacity = 0.4] (4,0) circle (2cm);
  \draw (4,0) circle (2cm);
  \draw (2,0) arc (180:360:2 and 0.6);
  \draw[dashed] (6,0) arc (0:180:2 and 0.6);
  \fill[fill=black] (6,0) circle (2pt);
  \fill[fill=black] (-1+4,-0.5) circle (2pt);
  \fill[fill=black] (1+4,-0.5) circle (2pt);
  \fill[fill=black] (-1+4,0.5) circle (2pt);
  \fill[fill=black] (1+4,0.5) circle (2pt);
  \shade[ball color = gray!40, opacity = 0.4] (8,0) circle (2cm);
  \draw (8,0) circle (2cm);
  \draw (6,0) arc (180:360:2 and 0.6);
  \draw[dashed] (10,0) arc (0:180:2 and 0.6);
  \fill[fill=black] (10,0) circle (2pt);
  \fill[fill=black] (-1+8,-0.5) circle (2pt);
  \fill[fill=black] (1+8,-0.5) circle (2pt);
  \fill[fill=black] (-1+8,0.5) circle (2pt);
  \fill[fill=black] (1+8,0.5) circle (2pt);
  \shade[ball color = gray!40, opacity = 0.4] (12,0) circle (2cm);
  \draw (12,0) circle (2cm);
  \draw (10,0) arc (180:360:2 and 0.6);
  \draw[dashed] (14,0) arc (0:180:2 and 0.6);
  \fill[fill=black] (14,0) circle (2pt);
  \fill[fill=black] (-1+12,-0.5) circle (2pt);
  \fill[fill=black] (1+12,-0.5) circle (2pt);
  \fill[fill=black] (-1+12,0.5) circle (2pt);
  \fill[fill=black] (1+12,0.5) circle (2pt);
  
  \shade[ball color = gray!40, opacity = 0.4] (2,-0.5*2) circle (2cm);
  \draw (2,-0.5*2) circle (2cm);
  \draw (0,-0.5*2) arc (180:360:2 and 0.6);
  \draw[dashed] (4,-0.5*2)  arc (0:180:2 and 0.6);
    \fill[fill=black] (1,-1.5) circle (2pt);
  \fill[fill=black] (3,-1.5) circle (2pt);
    \fill[fill=black] (0,-1) circle (2pt);
  \fill[fill=black] (4,-1) circle (2pt);
\end{tikzpicture}
\caption{Five spherical particles forming a simple structure (circle: port of the particles).}
\label{sphere}
\end{center}
\end{figure}

Our contribution is the following: we introduce a leader election algorithm based on local computations and simple to implement. This algorithm works when the structure the particles form has no hole (see Section 3). Also, since the algorithm can be described as a sequence of local computations, its limits (message complexity, required memory-space, etc) are easy to analyze.
We present a distributed algorithm to construct a spanning tree in the context of programmable matter and, also, a distributed algorithm to re-organize the port numbers of the particles.
Finally, we present an algorithm to assign a $k$-local identifier to each particle. In order to compute $k$-local identifiers, we suppose that we have done a leader election before. The $k$-local identifiers are determined using graph theoretical results about the coloring of the $k^{\text{th}}$ power of the grids.
An advantage of the given $k$-local identifiers is that they are really simple to update in case the particles move and, consequently, the structure that the particles form changes.

This paper is organized as follows: in Section 2, we present our algorithmic framework in the context of distributed algorithms for programmable network. In the third section, we present our leader election algorithm. Finally, in Section 4, we present our algorithm to assign $k$-local identifiers to the particles (using the colorings from Appendix~\ref{appendixd}).

\section{Notation, definitions and our programmable matter algorithmic framework}

The geometric amoebot model~\cite{ZD2014,ZD2015,ZD2015b,ZD2016a,ZD2016b,ZD2017} aims to model the computations that can occur in the context of programmable matter. In this paper, we use an algorithmic framework inspired by the geometric amoebot model.
We assume that any structure the different particles can form is a subgraph of an infinite graph $G$.
In this graph, $V(G)$ represents all possible positions the particles can occupy and $E(G)$ represents possible connections between particles. The set $E(G)$ also represents the possible movements from a position to another position (for a particle). We suppose that two particles can bond each other, i.e., can communicate only in the case they are on adjacent positions. The two following paragraphs are dedicated to the notation and definitions we use for graphs.

For a graph $G$, we denote by $V(G)$ the \textit{vertex set} of $G$ and by $E(G)\subseteq V(G)\times V(G)$ the \textit{edge set} of $G$. 
We denote by $d_G(u,v)$, the usual distance between two vertices $u$ and $v$ in $G$. If we consider the distance in a subgraph $H$ of $G$, the distance will be denoted by $d_H(u,v)$. 
The \textit{diameter} of $G$, denoted by $\text{diam}(G)$, is $\max(\{ d_{G}(u,v)|\ u,v\in V(G)\})$. The set $N_G(u)=\{ v\in V(G) |\ uv \in E(G)\}$ is the set of \textit{neighbors} of $u$. By $\Delta(G)$, we denote the \emph{maximum degree} in $G$, i.e., the maximum cardinality of $N_{G}(u)$, for $u\in V(G)$.
Finally, we denote by $G[S]$, for $S\subseteq V(G)$, the subgaph induced by the vertices from $S$ and by $G-S$ the subgraph of $G$ induced by the vertices from $V(G)\setminus S$.

\begin{figure}[t]
\begin{center}
\begin{tikzpicture}[scale=1]
\foreach \x in {-0.5,0,0.5,...,3.5}
{
\draw (\x,0) -- (\x,2);
}
\foreach \y in {0,0.5, ...,2}
{
\draw (-0.5,\y) -- (3.5,\y);
}

\foreach \x in {2,2.5,3,...,6}
{
\draw (\x+2.5,0) -- (\x+2.5,2);
}
\foreach \y in {0,0.5, ...,2}
{
\draw (2+2.5,\y) -- (6+2.5,\y);
}
\foreach \z in {2,2.5,3,3.5,4}
{
\draw (\z+2.5,2) -- (\z+2.5+2,0);
}
\foreach \h in {0.5,1,1.5}
{
\draw (2+2.5,\h) -- (\h+2.5+2,0);
\draw (\h+4+2.5,2) -- (6+2.5,\h);
}

\node at (0,0.5) [circle,draw=black,fill=black,scale=0.5]{};
\node at (1,1.5) [circle,draw=black,fill=black,scale=0.5]{};
\node at (-0.1,0.8) {\tiny{0}};
\node at (-0.3,0.35) {\tiny{3}};
\node at (0.1,0.2) {\tiny{2}};
\node at (0.3,0.6) {\tiny{1}};
\node at (-0.1+1,0.8+1) {\tiny{2}};
\node at (-0.3+1,0.35+1) {\tiny{1}};
\node at (0.1+1,0.2+1) {\tiny{0}};
\node at (0.3+1,0.6+1) {\tiny{3}};

\node at (0+5,0.5) [circle,draw=black,fill=black,scale=0.5]{};
\node at (7,1.5) [circle,draw=black,fill=black,scale=0.5]{};
\node at (-0.1+5,0.8) {\tiny{1}};
\node at (-0.35+5,0.7) {\tiny{0}};
\node at (-0.25+5,0.4) {\tiny{5}};
\node at (0.1+5,0.2) {\tiny{4}};
\node at (0.35+5,0.25) {\tiny{3}};
\node at (0.25+5,0.58) {\tiny{2}};
\node at (-0.1+7,0.8+1) {\tiny{0}};
\node at (-0.35+7,0.7+1) {\tiny{5}};
\node at (-0.25+7,0.4+1) {\tiny{4}};
\node at (0.1+7,0.2+1) {\tiny{3}};
\node at (0.35+7,0.25+1) {\tiny{2}};
\node at (0.25+7,0.58+1) {\tiny{1}};

\node at (6.5-5,-0.2) {(a)};
\node at (6.5,-0.2) {(b)};

\end{tikzpicture}
\end{center}
\begin{center}
\vspace{0.3cm}
\begin{tikzpicture}

\foreach \x in {2,2.5,3,...,6}
{
\draw (\x+7.5,0) -- (\x+7.5,2);
}
\foreach \y in {0,0.5, ...,2}
{
\draw (2+7.5,\y) -- (6+7.5,\y);
}
\foreach \z in {2,2.5,3,3.5,4}
{
\draw (\z+7.5,2) -- (\z+7.5+2,0);
}
\foreach \h in {0.5,1,1.5}
{
\draw (2+7.5,\h) -- (\h+7.5+2,0);
\draw (\h+4+7.5,2) -- (6+7.5,\h);
}

\foreach \z in {4,4.5,5,5.5,6}
{
\draw (\z+7.5,2) -- (\z-2+7.5,0);
}
\foreach \h in {0.5,1,1.5}
{
\draw (6+7.5,\h) -- (-\h+6+7.5,0);
\draw (\h+2+7.5,2) -- (7.5+2,2-\h);
}

\node at (0+10,0.5) [circle,draw=black,fill=black,scale=0.5]{};
\node at (11.5,0.5) [circle,draw=black,fill=black,scale=0.5]{};
\node at (-0.06+10,0.8) {\tiny{1}};
\node at (-0.25+10,0.6) {\tiny{0}};
\node at (-0.3+10,0.41) {\tiny{7}};
\node at (-0.1+10,0.28) {\tiny{6}};
\node at (0.08+10,0.2) {\tiny{5}};
\node at (0.25+10,0.36) {\tiny{4}};
\node at (0.3+10,0.58) {\tiny{3}};
\node at (0.1+10,0.7) {\tiny{2}};

\node at (-0.06+11.5,0.8) {\tiny{4}};
\node at (-0.25+11.5,0.6) {\tiny{3}};
\node at (-0.3+11.5,0.41) {\tiny{2}};
\node at (-0.1+11.5,0.28) {\tiny{1}};
\node at (0.08+11.5,0.2) {\tiny{0}};
\node at (0.25+11.5,0.36) {\tiny{7}};
\node at (0.3+11.5,0.58) {\tiny{6}};
\node at (0.1+11.5,0.7) {\tiny{5}};
\node at (6.5+5,-0.2) {(c)};

\end{tikzpicture}
\end{center}
\caption{Subgraphs of the square (Figure \ref{portfig}.a), triangular (Figure \ref{portfig}.b) and king (Figure \ref{portfig}.c) grids, with the port numbers of two particles.}
\label{portfig}
\end{figure}
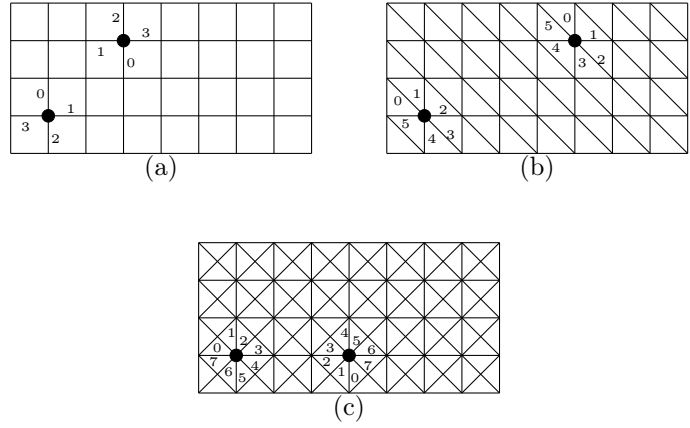

In the remaining part of this paper, the graphs considered will be the infinite \textit{square}, \textit{triangular} and \textit{king} grids. We denote by $\mathscr{S}$ the square grid, by $\mathscr{T}$ the triangular grid and by $\mathscr{K}$ the king grid.
A subgraph of each of these three infinite graphs is represented in Figure \ref{portfig}.
Moreover, we suppose that these three grids are represented on a plane as in Figure \ref{portfig}.
For these grids, the considered vertex set is $\{(i,j)|\ i,j \in \mathbb{Z}\}$ and the edge sets are the following:
\begin{itemize}
\item $E(\mathscr{S})=\{(i,j)(i\pm 1,j)|\ i,j \in \mathbb{Z} \}\cup \{(i,j)(i,j\pm 1)|\ i,j \in \mathbb{Z} \}$;
\item $E(\mathscr{T})=E(\mathcal{S})\cup \{(i,j)(i+ 1,j-1)|\ i,j \in \mathbb{Z} \}\cup \{(i,j)(i-1,j+ 1)|\ i,j \in \mathbb{Z} \}$;
\item $E(\mathscr{K})=E(\mathcal{T})\cup \{(i,j)(i+ 1,j+1)|\ i,j \in \mathbb{Z} \}\cup \{(i,j)(i-1,j- 1)|\ i,j \in \mathbb{Z} \}$. 
\end{itemize}
We also remind the distance between two vertices $(i,j)$ and $(i',j')$ in the three different grids: 
\begin{itemize}
\item $d_{\mathscr{S}}((i,j),(i',j'))=|i-i'|+|j-j'|$;
\item $d_{\mathscr{T}}((i,j),(i',j'))= \left\{
    \begin{array}{ll}
        \max(|i-i'|,|j-j'|), &\text{if } (i\ge i' \land j \le j' )\lor ( i\le i' \land j \ge j');\\
        |i-i'|+|j-j'|, &\text{otherwise;}
    \end{array}
\right.
$
\item $d_{\mathscr{K}}((i,j),(i',j'))=\max(|i-i'|,|j-j'|)$.
\end{itemize}
Note that there is a way to draw the triangular grid in which each triangle is equilateral. However, we prefer to draw it as a subgraph of the king grid (see Figure \ref{portfig}) in order to have illustrations for which the vertex set $\{(i,j)|\ i,j \in \mathbb{Z}\}$ corresponds to the position of the vertices in the plane. In both representation, the notion of distance coincide but is easier to observe in our chosen representation.
However, note that the representation of the triangular grid in which each triangle is equilateral corresponds to the optimal way to pack unit disks in the plane (the position of the vertices in this representation corresponds to the center of the unit disk and an edge represents a contact between two disks).

We also denote by $i\pmod{p}$ or $i_{\pmod p}$, depending on the context, the integer $j$ such that $j\equiv i \pmod{p}$ and $0\le j <p$. The remaining part of this subsection is dedicated to our programmable matter algorithmic framework.

We give the following properties about the particles and vertices of the graph:
\begin{itemize}
\item each particle occupies a single vertex and each vertex is occupied by at most one particle;
\item the subgraph induced by the occupied vertices is supposed to be connected.
\end{itemize} 
The subgraph induced by the occupied vertices of $V(G)$ is called the \textit{particle graph} and is denoted by $P$.
The vertex occupied by a particle $p$ is denoted by $s(p)$. 
For a particle $p$, $N_G(p)=\{u\in V(G)|\  u\in N_G(s(p))  \}$. 
The \textit{ports} of a particle are the endpoints of communication. Each particle has $\Delta(G)$ ports in a regular grid $G$ ($\Delta(G)=4$ for $G=\mathscr{S}$, $\Delta(G)=6$ for $G=\mathscr{T}$ and $\Delta(G)=8$ for $G=\mathscr{K}$). The ports of a particle occupying a vertex $u$ are represented by the edges incident with $u$. 
An edge between two vertices represents a possible communication between two particles $p_{1}$ and $p_{2}$ occupying these two vertices using each one a different port. 
A particle has the following properties:
\begin{itemize}
\item each particle is anonymous, i.e., it does not have an identifier;
\item each particle has a collection of ports, each labeled by a different integer from $\{0,\ldots,\Delta(G)-1\}$;
\item the port numbers are given as a function of the position of the edges on a plane representation of the grids (see Figure~\ref{portfig});
\item each particle knows the labels of the ports that can communicate with particles from the neighborhood;
\item each particle knows the state of the neighbors.
\end{itemize}

In our algorithmic framework, we suppose that the particles have their ports labeled following the same clockwise order. Thus, consecutive port numbers correspond to consecutive edges around a vertex (as in the representation on the plane from Figure~\ref{portfig}).
Note that the particles do not have the same notion of orientation, i.e., there is possibly not a unique label for ports that correspond to edges going in the same cardinal direction.
In the presented algorithms, the state of a particle will contain a variable corresponding to the status of the particle in the leader election algorithm and the information regarding its parents and childs for a constructed spanning tree.

The proposed algorithms in our algorithmic framework are results of successive local computations \cite{BAU2002,ROS1972}.
In particular, the first presented leader election algorithm from Section 3 can be described by a graph relabeling system \cite{BAU2002} which is a local computation system.
In this paper, the correct execution of the different algorithms is only guaranteed if the algorithms are ran in the order depicted in  Figure \ref{chapdepency}.

\begin{figure}[t]
\begin{center}
\begin{tikzpicture}[scale=0.5]
\node at (0,0){\fbox{\begin{minipage}{2.2cm} Leader election (Algorithm~\ref{alg1}) \end{minipage}}};
\node at (6,0){\fbox{\begin{minipage}{2.2cm} Spanning tree (Algorithm~\ref{alg2}) \end{minipage}}};
\node at (12.3,0){\fbox{\begin{minipage}{2.5cm} Port renumbering (Algorithm~\ref{alg3}) \end{minipage}}};
\node at (18.8,0){\fbox{\begin{minipage}{2.5cm} $k$-local identifiers (Algorithm~\ref{alg4}) \end{minipage}}};
\draw[->, thick] (2.4,0) -- (3.4,0);
\draw[->, thick] (8.4,0) -- (9.4,0);
\draw[->, thick] (15,0) -- (16,0);
\node at (18,-2.4){Graph coloring patterns};
\draw[->, thick] (18,-2) -- (18,-1);
\end{tikzpicture}
\end{center}
\caption{An illustration of the algorithm dependency (arrow between algorithms/results: dependency of one algorithm to another algorithm/result).}
\label{chapdepency}
\end{figure}
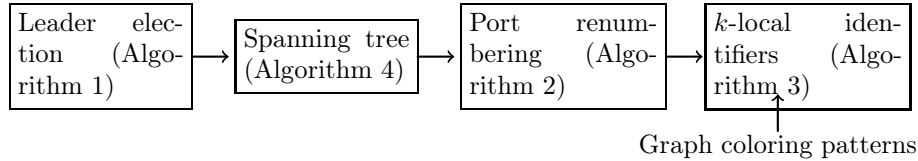

We suppose the following:

\begin{itemize}
\item each particle contains the same program and begins in the same state;
\item the computation process is represented by successive local computations;
\item no local computation occurs simultaneously on two particles at distance at most 2;
\item during a local computation, a particle can perform a bounded number of computations and can send messages to its neighbors;
\item a \emph{round} is a sequence of successive local computations for which each particle does at least one local  computation;
\item an algorithm finishes in $k$ rounds if after any $k$ successive rounds the algorithm is finished.
\end{itemize}
Note that the concept of rounds is used to bound the running time of the algorithms. 
In our algorithm framework we suppose that no two particles at distance at most $2$ perform computations simultaneously in order to simplify the presentation of our results. However, this supposition can be removed by implementing, for example, a probabilistic leader election algorithm on the vertices at distance at most $2$ of one of the two vertices, i.e., by computing a random value on the vertices at distance $2$ and doing the local computation following the increasing order of the values.
In order to compute the running time of an algorithm in case of a specific programmable matter prototype, the complexity of the algorithm should be computed using the required number of rounds and the required running time in order to avoid that two particles at distance at most $2$ perform computations simultaneously.

\section{Leader election}
In this section we present a new leader election algorithm.
This algorithm is very easy to implement but requires that the particle graph has a specific structure. In this algorithm, the required memory space is constant, the messages have constant size, the required computation power of the particle has been optimized and the required number of rounds is less than $2n$ ($n$ being the number of particles).

A \textit{hole} in a subgraph $G'$ of a graph $G$ among the three grids is a subgraph $H$ of $G$ satisfying three properties:
\begin{enumerate}
\item[i)] $V(H)$ is finite, $H$ is connected and $|V(H)|\ge 1$;
\item[ii)] $V(H)\cap V(G')=\emptyset$;
\item[iii)] every vertex $u\in V(H)$ satisfies $N_G(u)\subseteq V(H)\cup V(G')$.
\end{enumerate}
Less formally, a subgraph $G'$ of one of the three grids contains a hole if there is a finite space only containing vertices from $V(G)\setminus V(G')$ which are surrounded by vertices of $G'$.
A hole containing three vertices is illustrated in the left part of Figure \ref{holefig}. We call $G'$ \emph{hole-free}, when $G'$ has no holes.

If the particle graph $P$ on $G$ is hole-free, then every particle $p$ which satisfies $|N_G(p)\cap V(P))|<\Delta(G)$ is at the geographical border of the shape of $P$. Moreover, we call the set of particles $p$ which satisfy $|N_G(p)\cap V(P))|<\Delta(G)$ and such that the vertices $N_G(p)- V(P)$ are not all in a hole of $P$, the \emph{border} of $P$. The right part of Figure \ref{holefig} illustrates the border of $P$.

In addition, for a particle $p$ occupying a vertex $(i,j)$ of the square grid, the four vertices $(i+1,j+1)$, $(i-1,j+1)$, $(i+1,j-1)$ and $(i-1,j-1)$ are the \textit{corners} of $p$ and the set of corners is denoted by $C(p)$. The \emph{extended neighborhood} of a particle $p$, denoted by $M_G(p)$, is the set $N_{G}(p)$ if $G$ is the triangular grid or king grid or the set $N_{G}(p)\cup C(p)$ if $G$ is the square grid. Note that we define the extended neighborhood differently for the square grid in order to be able to present a generic algorithm (Algorithm \ref{alg1}) that works for all the three grids.

We give the following definition of $S$-contractible particle (see Figure \ref{figstmac}) that will be used in our leader election algorithm.

\begin{figure}[t]
\begin{center}
\begin{tikzpicture}[scale=1]

\foreach \x in {2,2.5,3,...,6}
{
\draw (\x+2.5,0) -- (\x+2.5,2);
}
\foreach \y in {0,0.5, ...,2}
{
\draw (2+2.5,\y) -- (6+2.5,\y);
}
\foreach \z in {2,2.5,3,3.5,4}
{
\draw (\z+2.5,2) -- (\z+2.5+2,0);
}
\foreach \h in {0.5,1,1.5}
{
\draw (2+2.5,\h) -- (\h+2.5+2,0);
\draw (\h+4+2.5,2) -- (6+2.5,\h);
}

\foreach \x in {2,2.5,3,...,6}
{
\draw (\x+2.5+5,0) -- (\x+2.5+5,2);
}
\foreach \y in {0,0.5, ...,2}
{
\draw (2+2.5+5,\y) -- (6+2.5+5,\y);
}
\foreach \z in {2,2.5,3,3.5,4}
{
\draw (\z+2.5+5,2) -- (\z+2.5+2+5,0);
}
\foreach \h in {0.5,1,1.5}
{
\draw (2+2.5+5,\h) -- (\h+2.5+2+5,0);
\draw (\h+4+2.5+5,2) -- (6+2.5+5,\h);
}

\node at (6,0.5) [circle,draw=black,fill=black,scale=0.5]{};
\node at (5.5,1) [circle,draw=black,fill=black,scale=0.5]{};
\node at (5.5,1.5) [circle,draw=black,fill=black,scale=0.5]{};
\node at (6,1.5) [circle,draw=black,fill=black,scale=0.5]{};
\node at (6.5,0) [circle,draw=black,fill=black,scale=0.5]{};
\node at (7,0) [circle,draw=black,fill=black,scale=0.5]{};
\node at (7,0.5) [circle,draw=black,fill=black,scale=0.5]{};
\node at (7,1) [circle,draw=black,fill=black,scale=0.5]{};
\node at (6.5,1.5) [circle,draw=black,fill=black,scale=0.5]{};
\node at (7,0.5) [circle,draw=black,fill=black,scale=0.5]{};
\node at (7,1.5) [circle,draw=black,fill=black,scale=0.5]{};
\node at (7,2) [circle,draw=black,fill=black,scale=0.5]{};
\node at (7.5,0.5) [circle,draw=black,fill=black,scale=0.5]{};
\node at (7.5,1.5) [circle,draw=black,fill=black,scale=0.5]{};
\node at (7.5,1) [circle,draw=black,fill=black,scale=0.5]{};

\node at (13.5-2,0) [regular polygon, regular polygon sides=4,draw=black,fill=red,scale=0.5]{};
\node at (13-2,0.5) [regular polygon, regular polygon sides=4,draw=black,fill=red,scale=0.5]{};
\node at (13.5-2,0.5) [regular polygon, regular polygon sides=4,draw=black,fill=red,scale=0.5]{};
\node at (14-2,0.5) [regular polygon, regular polygon sides=4,draw=black,fill=red,scale=0.5]{};
\node at (12.5-2,1) [regular polygon, regular polygon sides=4,draw=black,fill=red,scale=0.5]{};
\node at (13-2,1) [circle,draw=black,fill=black,scale=0.5]{};
\node at (13.5-2,1) [circle,draw=black,fill=black,scale=0.5]{};
\node at (14-2,1) [regular polygon, regular polygon sides=4,draw=black,fill=red,scale=0.5]{};
\node at (14.5-2,1) [regular polygon, regular polygon sides=4,draw=black,fill=red,scale=0.5]{};
\node at (12.5-2,1.5) [regular polygon, regular polygon sides=4,draw=black,fill=red,scale=0.5]{};
\node at (13.5-2,1.5) [circle,draw=black,fill=black,scale=0.5]{};
\node at (13-2,1.5) [regular polygon, regular polygon sides=4,draw=black,fill=red,scale=0.5]{};
\node at (14-2,1.5) [circle,draw=black,fill=black,scale=0.5]{};
\node at (14.5-2,1.5) [regular polygon, regular polygon sides=4,draw=black,fill=red,scale=0.5]{};
\node at (15-2,1.5) [regular polygon, regular polygon sides=4,draw=black,fill=red,scale=0.5]{};
\node at (15.5-2,1.5) [regular polygon, regular polygon sides=4,draw=black,fill=red,scale=0.5]{};
\node at (13-2,2) [regular polygon, regular polygon sides=4,draw=black,fill=red,scale=0.5]{};
\node at (13.5-2,2) [regular polygon, regular polygon sides=4,draw=black,fill=red,scale=0.5]{};
\node at (14-2,2) [regular polygon, regular polygon sides=4,draw=black,fill=red,scale=0.5]{};
\node at (14.5-2,2) [regular polygon, regular polygon sides=4,draw=black,fill=red,scale=0.5]{};
\end{tikzpicture}
\end{center}
\caption{A hole in $P$ (on the left) and the border of $P$ in the case $P$ is hole-free (on the right; square: particle on the border of $P$).}
\label{holefig}
\end{figure}
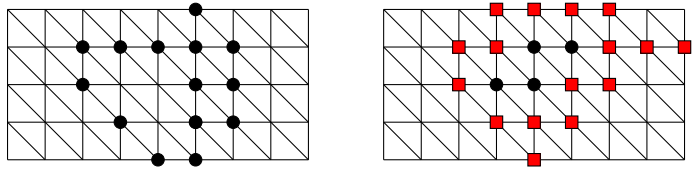
\begin{de}
Let $G$ be an infinite grid among $\mathscr{S}$, $\mathscr{T}$ and $\mathscr{K}$ and let $S\subseteq V(P)$, for $P$ the particle graph on $G$.
A particle $p$ is said to be \textit{$S$-contractible} if it satisfies the following properties:
\begin{enumerate}
\item[I)] $G[M_{G}(p)\cap S]$ is connected;
\item[II)] $|N_{G}(p)\cap S|<\Delta(G)$, i.e., there exists a neighbor of $p$ in $G$ which is not occupied by a particle from $S$.
\end{enumerate}
\end{de}

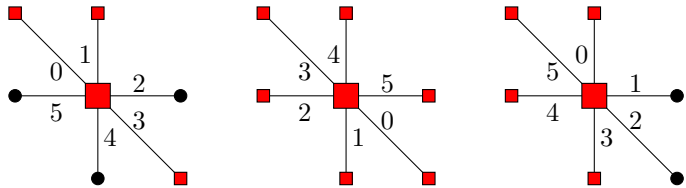
\begin{figure}
\begin{center}
\begin{tikzpicture}[scale=1.1]
\draw (0,-1) -- (0,1);
\draw (-1,0) -- (1,0);
\draw (1,-1) -- (-1,1);

\draw (0+3,-1) -- (0+3,1);
\draw (-1+3,0) -- (1+3,0);
\draw (1+3,-1) -- (-1+3,1);

\draw (0+6,-1) -- (0+6,1);
\draw (-1+6,0) -- (1+6,0);
\draw (1+6,-1) -- (-1+6,1);

\node at (0,0)  [regular polygon, regular polygon sides=4,draw=black,fill=red,scale=1]{};
\node at (0,1)  [regular polygon, regular polygon sides=4,draw=black,fill=red,scale=0.5]{};
\node at (-1,1) [regular polygon, regular polygon sides=4,draw=black,fill=red,scale=0.5]{};
\node at (-1,0) [circle,draw=black,fill=black,scale=0.5]{};
\node at (1,0) [circle,draw=black,fill=black,scale=0.5]{};
\node at (1,-1)  [regular polygon, regular polygon sides=4,draw=black,fill=red,scale=0.5]{};
\node at (0,-1) [circle,draw=black,fill=black,scale=0.5]{};
\node at (-0.5,0.3){$0$};
\node at (-0.15,0.5){$1$};
\node at (0.5,0.15){$2$};
\node at (0.5,-0.3){$3$};
\node at (0.15,-0.5){$4$};
\node at (-0.5,-0.2){$5$};

\node at (0+3,0)  [regular polygon, regular polygon sides=4,draw=black,fill=red,scale=1]{};
\node at (0+3,1)  [regular polygon, regular polygon sides=4,draw=black,fill=red,scale=0.5]{};
\node at (-1+3,1) [regular polygon, regular polygon sides=4,draw=black,fill=red,scale=0.5]{};
\node at (-1+3,0) [regular polygon, regular polygon sides=4,draw=black,fill=red,scale=0.5]{};
\node at (1+3,0) [regular polygon, regular polygon sides=4,draw=black,fill=red,scale=0.5]{};
\node at (1+3,-1)  [regular polygon, regular polygon sides=4,draw=black,fill=red,scale=0.5]{};
\node at (0+3,-1) [regular polygon, regular polygon sides=4,draw=black,fill=red,scale=0.5]{};
\node at (-0.5+3,0.3){$3$};
\node at (-0.15+3,0.5){$4$};
\node at (0.5+3,0.15){$5$};
\node at (0.5+3,-0.3){$0$};
\node at (0.15+3,-0.5){$1$};
\node at (-0.5+3,-0.2){$2$};

\node at (0+6,0)  [regular polygon, regular polygon sides=4,draw=black,fill=red,scale=1]{};
\node at (0+6,1)  [regular polygon, regular polygon sides=4,draw=black,fill=red,scale=0.5]{};
\node at (-1+6,1) [regular polygon, regular polygon sides=4,draw=black,fill=red,scale=0.5]{};
\node at (-1+6,0) [regular polygon, regular polygon sides=4,draw=black,fill=red,scale=0.5]{};
\node at (1+6,0) [circle,draw=black,fill=black,scale=0.5]{};
\node at (1+6,-1)  [circle,draw=black,fill=black,scale=0.5]{};
\node at (0+6,-1) [regular polygon, regular polygon sides=4,draw=black,fill=red,scale=0.5]{};
\node at (-0.5+6,0.3){$5$};
\node at (-0.15+6,0.5){$0$};
\node at (0.5+6,0.15){$1$};
\node at (0.5+6,-0.3){$2$};
\node at (0.15+6,-0.5){$3$};
\node at (-0.5+6,-0.2){$4$};
\end{tikzpicture}
\end{center}
\caption{Two non $S$-contractible particles (at the center of the left and the middle drawing) and an $S$-contractible particle (at the center of the right drawing) in the triangular grid (square: particle in $S$; circle: particle not in $S$).}
\label{figstmac}
\end{figure}

A particle $p$ is an \emph{articulation} of a connected subgraph $G'$ of one of the three grids if $G'-\{s(p)\}$ is not connected.
Derakhshandeh et al. \cite{ZD2015b} proposed a randomized leader election algorithm in the geometric amoebot model in the case there is no particle which is an articulation.
Our proposed leader election algorithm (Algorithm \ref{alg1}) works even if $V(P)$ contains a particle $p$ which is an articulation.
However, in contrast with the leader election algorithm from Derakhshandeh et al. \cite{ZD2015b}, Algorithm \ref{alg1} does not work if $P$ has holes. In the remaining part of this paper, Algorithm \ref{alg1} is called the $S$-contraction algorithm.

Recently, Daymude et al. \cite{JJD2017} have improved the algorithm from Derakhshandeh et al. \cite{ZD2015b} in order that it works when $V(P)$ contains an articulation. However, it remains challenging to implement it.

Also, very recently, Di Luna et al. \cite{DIL2017} have introduced a leader election algorithm called consumption algorithm. The consumption and the $S$-contraction algorithms both consist in successively removing the candidacy of the particles on the border of $P$.
However, one can easily notice that, in our algorithm, we possibly remove the candidacy of particles having four or five neighbors (which is not considered in the consumption algorithm).  Also, the consumption algorithm  does not work on square and king grids and the considered theoretical frameworks for the two algorithms are different.

In the $S$-contraction algorithm (Algorithm \ref{alg1}), the particles can be in three different states: \textbf{C} (candidate), \textbf{N} (not elected) and \textbf{L} (leader). We suppose that every particle begins in the state \textbf{C}.

%
%
%

Let $S$ be the particle in state \textbf{C}. Algorithm \ref{alg1} consists in removing from $S$ the particles which are both on the border of $G[S]$ and not articulations of $G[S]$. 
An example of the execution of Algorithm \ref{alg1} is illustrated by Figure~\ref{proccess}. Note that, depending on the order in which the local computations occur, the result of the execution of the algorithm could be different. 
For example, between the configuration of Figure~\ref{proccess}.c and that of Figure~\ref{proccess}.d, we suppose that the local computations occur in this order: first a local computation occurs for the bottom left particle, second it occurs for the upper left particle, third it occurs for the upper right particle and fourth it occurs for the last particle (we only consider the particles which are in state \textbf{C}).

\begin{algorithm}
\caption{The $S$-contraction algorithm for a particle $p$ and $S$ the set of particles in state \textbf{C}.} 
\label{alg1}
\begin{algorithmic} 
\State \textbf{Case 1: } State \textbf{C}.
\If {the particle is $S$-contractible}
    \If {the particle has no neighbor in $S$}
        \State set the state to \textbf{L}.
    \Else
        \State set the state to \textbf{N}.
    \EndIf
\Else
        \State  stay in state \textbf{C}.
\EndIf
\State \textbf{Case 2: } States \textbf{L} or \textbf{N}. \\ Perform no further actions. 
\end{algorithmic}
\end{algorithm}

\begin{figure}[t]
\begin{center}
\begin{tikzpicture}[scale=1]

\foreach \x in {2,2.5,3,...,6}
{
\draw (\x+2.5,0) -- (\x+2.5,2);
}
\foreach \y in {0,0.5, ...,2}
{
\draw (2+2.5,\y) -- (6+2.5,\y);
}
\foreach \z in {2,2.5,3,3.5,4}
{
\draw (\z+2.5,2) -- (\z+2.5+2,0);
}
\foreach \h in {0.5,1,1.5}
{
\draw (2+2.5,\h) -- (\h+2.5+2,0);
\draw (\h+4+2.5,2) -- (6+2.5,\h);
}

\foreach \x in {2,2.5,3,...,6}
{
\draw (\x+2.5+5,0) -- (\x+2.5+5,2);
}
\foreach \y in {0,0.5, ...,2}
{
\draw (2+2.5+5,\y) -- (6+2.5+5,\y);
}
\foreach \z in {2,2.5,3,3.5,4}
{
\draw (\z+2.5+5,2) -- (\z+2.5+2+5,0);
}
\foreach \h in {0.5,1,1.5}
{
\draw (2+2.5+5,\h) -- (\h+2.5+2+5,0);
\draw (\h+4+2.5+5,2) -- (6+2.5+5,\h);
}

\node at (5,0.5) [regular polygon, regular polygon sides=4,draw=black,fill=red,scale=0.5]{};
\node at (5.5,0.5) [circle,draw=black,fill=black,scale=0.5]{};
\node at (6,0.5) [circle,draw=black,fill=black,scale=0.5]{};
\node at (6,1) [circle,draw=black,fill=black,scale=0.5]{};
\node at (6.5,0.5) [regular polygon, regular polygon sides=4,draw=black,fill=red,scale=0.5]{};
\node at (6.5,1) [regular polygon, regular polygon sides=4,draw=black,fill=red,scale=0.5]{};
\node at (7,0.5) [regular polygon, regular polygon sides=4,draw=black,fill=red,scale=0.5]{};
\node at (7,1) [regular polygon, regular polygon sides=4,draw=black,fill=red,scale=0.5]{};
\node at (5.5,1.5) [circle,draw=black,fill=black,scale=0.5]{};
\node at (5,2) [regular polygon, regular polygon sides=4,draw=black,fill=red,scale=0.5]{};
\node at (7.5,1.5) [regular polygon, regular polygon sides=4,draw=black,fill=red,scale=0.5]{};
\node at (7.5,0.5) [regular polygon, regular polygon sides=4,draw=black,fill=red,scale=0.5]{};
\node at (7.5,1) [circle,draw=black,fill=black,scale=0.5]{};
\node at (8,0.5) [regular polygon, regular polygon sides=4,draw=black,fill=red,scale=0.5]{};

\node at (5+5,0.5) [regular polygon, regular polygon sides=3,draw=black,fill=green,scale=0.4]{};
\node at (5.5+5,0.5) [regular polygon, regular polygon sides=4,draw=black,fill=red,scale=0.5]{};
\node at (6+5,0.5) [circle,draw=black,fill=black,scale=0.5]{};
\node at (6+5,1) [circle,draw=black,fill=black,scale=0.5]{};
\node at (6.5+5,0.5) [regular polygon, regular polygon sides=3,draw=black,fill=green,scale=0.4]{};
\node at (6.5+5,1)  [circle,draw=black,fill=black,scale=0.5]{};
\node at (7+5,0.5) [regular polygon, regular polygon sides=3,draw=black,fill=green,scale=0.4]{};
\node at (7+5,1)  [circle,draw=black,fill=black,scale=0.5]{};
\node at (5.5+5,1.5) [regular polygon, regular polygon sides=4,draw=black,fill=red,scale=0.5]{};
\node at (5+5,2) [regular polygon, regular polygon sides=3,draw=black,fill=green,scale=0.4]{};
\node at (7.5+5,1.5) [regular polygon, regular polygon sides=3,draw=black,fill=green,scale=0.4]{};
\node at (7.5+5,0.5) [regular polygon, regular polygon sides=3,draw=black,fill=green,scale=0.4]{};
\node at (7.5+5,1) [regular polygon, regular polygon sides=4,draw=black,fill=red,scale=0.5]{};
\node at (8+5,0.5) [regular polygon, regular polygon sides=3,draw=black,fill=green,scale=0.4]{};

\node at (6.5,-0.2) {(a)};
\node at (6.5+5,-0.2) {(b)};

\end{tikzpicture}
\end{center}
\begin{center}
\begin{tikzpicture}[scale=1]

\foreach \x in {2,2.5,3,...,6}
{
\draw (\x+2.5,0) -- (\x+2.5,2);
}
\foreach \y in {0,0.5, ...,2}
{
\draw (2+2.5,\y) -- (6+2.5,\y);
}
\foreach \z in {2,2.5,3,3.5,4}
{
\draw (\z+2.5,2) -- (\z+2.5+2,0);
}
\foreach \h in {0.5,1,1.5}
{
\draw (2+2.5,\h) -- (\h+2.5+2,0);
\draw (\h+4+2.5,2) -- (6+2.5,\h);
}

\foreach \x in {2,2.5,3,...,6}
{
\draw (\x+2.5+5,0) -- (\x+2.5+5,2);
}
\foreach \y in {0,0.5, ...,2}
{
\draw (2+2.5+5,\y) -- (6+2.5+5,\y);
}
\foreach \z in {2,2.5,3,3.5,4}
{
\draw (\z+2.5+5,2) -- (\z+2.5+2+5,0);
}
\foreach \h in {0.5,1,1.5}
{
\draw (2+2.5+5,\h) -- (\h+2.5+2+5,0);
\draw (\h+4+2.5+5,2) -- (6+2.5+5,\h);
}

\node at (5,0.5) [regular polygon, regular polygon sides=3,draw=black,fill=green,scale=0.4]{};
\node at (5.5,0.5) [regular polygon, regular polygon sides=3,draw=black,fill=green,scale=0.4]{};
\node at (6,0.5) [regular polygon, regular polygon sides=4,draw=black,fill=red,scale=0.5]{};
\node at (6,1) [circle,draw=black,fill=black,scale=0.5]{};
\node at (6.5,0.5) [regular polygon, regular polygon sides=3,draw=black,fill=green,scale=0.4]{};
\node at (6.5,1)  [circle,draw=black,fill=black,scale=0.5]{};
\node at (7,0.5) [regular polygon, regular polygon sides=3,draw=black,fill=green,scale=0.4]{};
\node at (7,1)  [regular polygon, regular polygon sides=4,draw=black,fill=red,scale=0.5]{};
\node at (5.5,1.5) [regular polygon, regular polygon sides=3,draw=black,fill=green,scale=0.4]{};
\node at (5,2) [regular polygon, regular polygon sides=3,draw=black,fill=green,scale=0.4]{};
\node at (7.5,1.5) [regular polygon, regular polygon sides=3,draw=black,fill=green,scale=0.4]{};
\node at (7.5,0.5) [regular polygon, regular polygon sides=3,draw=black,fill=green,scale=0.4]{};
\node at (7.5,1) [regular polygon, regular polygon sides=3,draw=black,fill=green,scale=0.4]{};
\node at (8,0.5) [regular polygon, regular polygon sides=3,draw=black,fill=green,scale=0.4]{};

\node at (5+5,0.5) [regular polygon, regular polygon sides=3,draw=black,fill=green,scale=0.4]{};
\node at (5.5+5,0.5) [regular polygon, regular polygon sides=3,draw=black,fill=green,scale=0.4]{};
\node at (6+5,0.5) [regular polygon, regular polygon sides=3,draw=black,fill=green,scale=0.4]{};
\node at (6+5,1) [regular polygon, regular polygon sides=3,draw=black,fill=green,scale=0.4]{};
\node at (6+5.5,0.5) [regular polygon, regular polygon sides=3,draw=black,fill=green,scale=0.4]{};
\node at (6.5+5,1)  [regular polygon, regular polygon sides=5,draw=black,fill=blue,scale=0.5]{};
\node at (7+5,0.5) [regular polygon, regular polygon sides=3,draw=black,fill=green,scale=0.4]{};
\node at (7+5,1)  [regular polygon, regular polygon sides=3,draw=black,fill=green,scale=0.4]{};
\node at (5.5+5,1.5) [regular polygon, regular polygon sides=3,draw=black,fill=green,scale=0.4]{};
\node at (5+5,2) [regular polygon, regular polygon sides=3,draw=black,fill=green,scale=0.4]{};
\node at (7.5+5,1.5) [regular polygon, regular polygon sides=3,draw=black,fill=green,scale=0.4]{};
\node at (7.5+5,0.5) [regular polygon, regular polygon sides=3,draw=black,fill=green,scale=0.4]{};
\node at (7.5+5,1) [regular polygon, regular polygon sides=3,draw=black,fill=green,scale=0.4]{};
\node at (8+5,0.5) [regular polygon, regular polygon sides=3,draw=black,fill=green,scale=0.4]{};

\node at (6.5,-0.2) {(c)};
\node at (6.5+5,-0.2) {(d)};
\end{tikzpicture}
\end{center}
\caption{An example of the execution of $S$-contraction algorithm after one round (Figure~\ref{proccess}.a), two rounds (Figure~\ref{proccess}.b), three rounds (Figure~\ref{proccess}.c) and after four rounds (Figure~\ref{proccess}.d; circle: non $S$-contractible particle; square: $S$-contractible particle; triangle: particle in state \textbf{N}; pentagon: particle in state \textbf{L}; $S$ being the set of particles in state \textbf{C}).}
\label{proccess}
\end{figure}
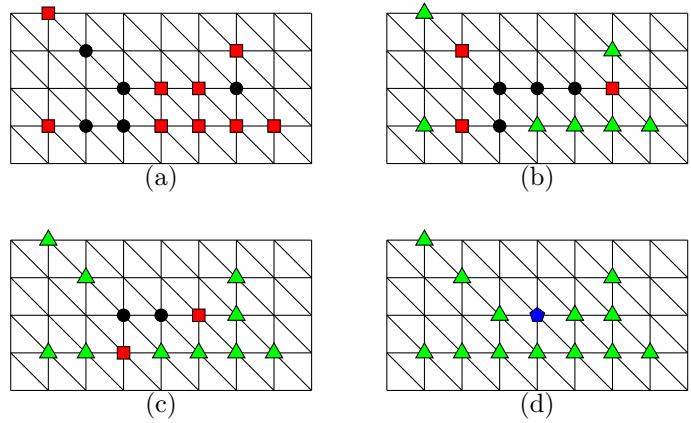

\begin{theo}\label{theobavure}
Let $S$ be the set of particles in state \textbf{C} and $P$ be the particle graph on $G$.
If $P$ is hole-free, then at the end of the execution of the $S$-contraction algorithm, there will be exactly one particle in the state $L$.
\end{theo}

In Appendix \ref{appendixa}, the proof of Theorem \ref{theobavure} is given. Also, a bound on the complexity of the $S$-contraction algorithm is given. In Appendix \ref{appendixc}, it is explained how to combine the $S$-contraction algorithm with a general leader election algorithm.

\section{Assigning $k$-local identifiers to particles}
In this section, we combine the results from Section 3 and Appendix \ref{appendixd} in order to correctly compute a $k$-local identifier. In a first subsection, we describe a way to create a spanning tree of particles and a way to change the ports numbering of the different particles.
In a second subsection, we describe how to compute $k$-local identifiers based on the coloring functions from Appendix \ref{appendixd}.

We suppose that Algorithms~\ref{alg3} and~\ref{alg4} are preceded by a leader election algorithm (which could be Algorithm~\ref{alg1}). Then it follows that there is a single particle in a specific state (the leader) and all the remaining particles are in the same state (non elected).

\subsection{Re-organizing the particles}

By $N^{+}_{G}(u)$ we denote the set of port numbers which can communicate with particles occupying vertices from $N_{G}(u)$.
When there is a leader, we can easily compute a spanning tree using a distributed algorithm (see Appendix \ref{appendixb}).
Now suppose that for each particle $p$, we have two set of ports $\text{parent}(p)$ and $\text{child}(p)$ which contains the port numbers of the particles in communication with its parent and with its children, respectively, in the spanning tree. In this way, the required memory in order to store where are the children and the parent of the particle in a spanning tree is constant (since the maximum degree is bounded in the considered grids).

In our proposed Algorithm~\ref{alg3}, the goal is to change the way the port are numbered in order that every particle has its ports numbered by the same number going in the same cardinal direction in the different grids. This algorithm does not work if we do not have a leader among the different particles. The function $r_{G}$ used in Algorithm~\ref{alg3} is defined, depending the choice of $G$, as follows: $r_{\mathscr{S}}(i)=(i+2)\pmod{4}$, $r_{\mathscr{T}}(i)=(i+3)\pmod{6}$ and $r_{\mathscr{K}}(i)=(i+3)\pmod{6}$.
\begin{algorithm}
\caption{The port renumbering algorithm for a particle $p$.} 
\label{alg3}
\begin{algorithmic} 
\State \textbf{Case 1}: State \textbf{L}.
\State for each port $a$ from $\text{child}(p)$ send a message $m_a$, containing $a$, through port $a$.
\State \textbf{Case 2:} State \textbf{N}.
\If { the particle receives the message $m_b$, containing $b$, through the port $a$}
    \State change the port number $a$ to $r_{G}(b)$ and changes the port numbers of the other ports following the clockwise order;
    \State update both $\text{parent}(p)$ and $\text{child}(p)$.
\EndIf
\end{algorithmic}
\end{algorithm}

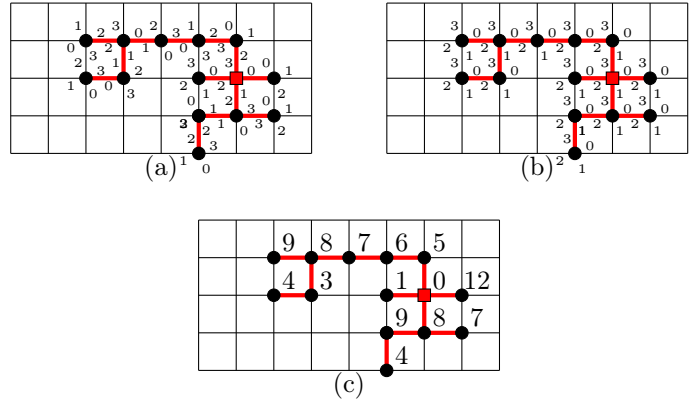
\begin{figure}[t]
\begin{center}
\begin{tikzpicture}[scale=1]

\foreach \x in {2,2.5,3,...,6}
{
\draw (\x+2.5,0) -- (\x+2.5,2);
}
\foreach \y in {0,0.5, ...,2}
{
\draw (2+2.5,\y) -- (6+2.5,\y);
}

\foreach \x in {2,2.5,3,...,6}
{
\draw (\x+2.5+5,0) -- (\x+2.5+5,2);
}
\foreach \y in {0,0.5, ...,2}
{
\draw (2+2.5+5,\y) -- (6+2.5+5,\y);
}

\draw[ultra thick, color=red] (5.5,1) -- (6,1);
\draw[ultra thick, color=red] (5.5,1.5) -- (7.5,1.5);
\draw[ultra thick, color=red] (6,1.5) -- (6,1);
\draw[ultra thick, color=red] (7.5,1.5) -- (7.5,0.5);
\draw[ultra thick, color=red] (7,1) -- (8,1);
\draw[ultra thick, color=red] (7,0.5) -- (8,0.5);
\draw[ultra thick, color=red] (7,0.5) -- (7,0);

\draw[ultra thick, color=red] (5.5+5,1) -- (6+5,1);
\draw[ultra thick, color=red] (5.5+5,1.5) -- (7.5+5,1.5);
\draw[ultra thick, color=red] (6+5,1.5) -- (6+5,1);
\draw[ultra thick, color=red] (7.5+5,1.5) -- (7.5+5,0.5);
\draw[ultra thick, color=red] (7+5,1) -- (8+5,1);
\draw[ultra thick, color=red] (7+5,0.5) -- (8+5,0.5);
\draw[ultra thick, color=red] (7+5,0.5) -- (7+5,0);

\node at (6,1) [circle,draw=black,fill=black,scale=0.5]{};
\node at (5.5,1) [circle,draw=black,fill=black,scale=0.5]{};
\node at (5.5,1.5) [circle,draw=black,fill=black,scale=0.5]{};
\node at (6,1.5) [circle,draw=black,fill=black,scale=0.5]{};
\node at (8,0.5) [circle,draw=black,fill=black,scale=0.5]{};
\node at (7,0) [circle,draw=black,fill=black,scale=0.5]{};
\node at (7,0.5) [circle,draw=black,fill=black,scale=0.5]{};
\node at (8,1) [circle,draw=black,fill=black,scale=0.5]{};
\node at (6.5,1.5) [circle,draw=black,fill=black,scale=0.5]{};
\node at (7,0.5) [circle,draw=black,fill=black,scale=0.5]{};
\node at (7,1.5) [circle,draw=black,fill=black,scale=0.5]{};
\node at (7,1) [circle,draw=black,fill=black,scale=0.5]{};
\node at (7.5,0.5) [circle,draw=black,fill=black,scale=0.5]{};
\node at (7.5,1.5) [circle,draw=black,fill=black,scale=0.5]{};
\node at (7.5,1) [regular polygon, regular polygon sides=4,draw=black,fill=red,scale=0.5]{};

\node at (6.2,1.1) []{\tiny{2}};
\node at (5.7,1.1) []{\tiny{3}};
\node at (5.7,1.6) []{\tiny{2}};
\node at (6.2,1.6) []{\tiny{0}};
\node at (8.2,0.6) []{\tiny{1}};
\node at (7.2,0.1) []{\tiny{3}};
\node at (8.2,1.1) []{\tiny{1}};
\node at (6.7,1.6) []{\tiny{3}};
\node at (7.2,0.6) []{\tiny{1}};
\node at (7.2,1.6) []{\tiny{2}};
\node at (7.2,1.1) []{\tiny{0}};
\node at (7.7,0.6) []{\tiny{3}};
\node at (7.7,1.6) []{\tiny{1}};
\node at (7.7,1.1) []{\tiny{0}};

\node at (6.1,0.8) []{\tiny{3}};
\node at (5.6,0.8) []{\tiny{0}};
\node at (5.6,1.3) []{\tiny{3}};
\node at (6.1,1.3) []{\tiny{1}};
\node at (8.1,0.3) []{\tiny{2}};
\node at (7.1,-0.2) []{\tiny{0}};
\node at (8.1,0.8) []{\tiny{2}};
\node at (6.6,1.3) []{\tiny{0}};
\node at (7.1,0.3) []{\tiny{2}};
\node at (7.1,1.3) []{\tiny{3}};
\node at (7.1,0.8) []{\tiny{1}};
\node at (7.6,0.3) []{\tiny{0}};
\node at (7.6,1.3) []{\tiny{2}};
\node at (7.6,0.8) []{\tiny{1}};

\node at (5.8,0.9) []{\tiny{0}};
\node at (5.3,0.9) []{\tiny{1}};
\node at (5.3,1.4) []{\tiny{0}};
\node at (5.8,1.4) []{\tiny{2}};
\node at (7.8,0.4) []{\tiny{3}};
\node at (6.8,-0.1) []{\tiny{1}};
\node at (6.8,0.4) []{\tiny{2}};
\node at (6.3,1.4) []{\tiny{1}};
\node at (6.8,0.4) []{\tiny{3}};
\node at (6.8,1.4) []{\tiny{0}};
\node at (6.8,0.9) []{\tiny{2}};
\node at (7.3,0.4) []{\tiny{1}};
\node at (7.3,1.4) []{\tiny{3}};
\node at (7.3,0.9) []{\tiny{2}};

\node at (5.9,1.2) []{\tiny{1}};
\node at (5.4,1.2) []{\tiny{2}};
\node at (5.4,1.7) []{\tiny{1}};
\node at (5.9,1.7) []{\tiny{3}};
\node at (7.9,0.7) []{\tiny{0}};
\node at (6.9,0.2) []{\tiny{2}};
\node at (7.9,1.2) []{\tiny{0}};
\node at (6.4,1.7) []{\tiny{2}};
\node at (6.9,0.7) []{\tiny{0}};
\node at (6.9,1.7) []{\tiny{1}};
\node at (6.9,1.2) []{\tiny{3}};
\node at (7.4,0.7) []{\tiny{2}};
\node at (7.4,1.7) []{\tiny{0}};
\node at (7.4,1.2) []{\tiny{3}};

\node at (6+5,1) [circle,draw=black,fill=black,scale=0.5]{};
\node at (5.5+5,1) [circle,draw=black,fill=black,scale=0.5]{};
\node at (5.5+5,1.5) [circle,draw=black,fill=black,scale=0.5]{};
\node at (6+5,1.5) [circle,draw=black,fill=black,scale=0.5]{};
\node at (8+5,0.5) [circle,draw=black,fill=black,scale=0.5]{};
\node at (7+5,0) [circle,draw=black,fill=black,scale=0.5]{};
\node at (7+5,0.5) [circle,draw=black,fill=black,scale=0.5]{};
\node at (8+5,1) [circle,draw=black,fill=black,scale=0.5]{};
\node at (6.5+5,1.5) [circle,draw=black,fill=black,scale=0.5]{};
\node at (7+5,0.5) [circle,draw=black,fill=black,scale=0.5]{};
\node at (7+5,1.5) [circle,draw=black,fill=black,scale=0.5]{};
\node at (7+5,1) [circle,draw=black,fill=black,scale=0.5]{};
\node at (7.5+5,0.5) [circle,draw=black,fill=black,scale=0.5]{};
\node at (7.5+5,1.5) [circle,draw=black,fill=black,scale=0.5]{};
\node at (7.5+5,1) [regular polygon, regular polygon sides=4,draw=black,fill=red,scale=0.5]{};

\node at (6.2+5,1.1) []{\tiny{0}};
\node at (5.7+5,1.1) []{\tiny{0}};
\node at (5.7+5,1.6) []{\tiny{0}};
\node at (6.2+5,1.6) []{\tiny{0}};
\node at (8.2+5,0.6) []{\tiny{0}};
\node at (7.2+5,0.1) []{\tiny{0}};
\node at (7.2+5,0.6) []{\tiny{0}};
\node at (8.2+5,1.1) []{\tiny{0}};
\node at (6.7+5,1.6) []{\tiny{0}};
\node at (7.2+5,1.6) []{\tiny{0}};
\node at (7.2+5,1.1) []{\tiny{0}};
\node at (7.7+5,0.6) []{\tiny{0}};
\node at (7.7+5,1.6) []{\tiny{0}};
\node at (7.7+5,1.1) []{\tiny{0}};

\node at (6.1+5,0.8) []{\tiny{1}};
\node at (5.6+5,0.8) []{\tiny{1}};
\node at (5.6+5,1.3) []{\tiny{1}};
\node at (6.1+5,1.3) []{\tiny{1}};
\node at (8.1+5,0.3) []{\tiny{1}};
\node at (7.1+5,-0.2) []{\tiny{1}};
\node at (7.1+5,0.3) []{\tiny{1}};
\node at (8.1+5,0.8) []{\tiny{1}};
\node at (6.6+5,1.3) []{\tiny{1}};
\node at (7.1+5,0.3) []{\tiny{1}};
\node at (7.1+5,0.8) []{\tiny{1}};
\node at (7.6+5,0.3) []{\tiny{1}};
\node at (7.6+5,1.3) []{\tiny{1}};
\node at (7.6+5,0.8) []{\tiny{1}};

\node at (5.8+5,0.9) []{\tiny{2}};
\node at (5.3+5,0.9) []{\tiny{2}};
\node at (5.3+5,1.4) []{\tiny{2}};
\node at (5.8+5,1.4) []{\tiny{2}};
\node at (7.8+5,0.4) []{\tiny{2}};
\node at (6.8+5,-0.1) []{\tiny{2}};
\node at (6.8+5,0.4) []{\tiny{2}};
\node at (7.8+5,0.9) []{\tiny{2}};
\node at (6.3+5,1.4) []{\tiny{2}};
\node at (6.8+5,1.4) []{\tiny{2}};
\node at (6.8+5,0.9) []{\tiny{2}};
\node at (7.3+5,0.4) []{\tiny{2}};
\node at (7.3+5,1.4) []{\tiny{2}};
\node at (7.3+5,0.9) []{\tiny{2}};

\node at (5.9+5,1.2) []{\tiny{3}};
\node at (5.4+5,1.2) []{\tiny{3}};
\node at (5.4+5,1.7) []{\tiny{3}};
\node at (5.9+5,1.7) []{\tiny{3}};
\node at (7.9+5,0.7) []{\tiny{3}};
\node at (6.9+5,0.2) []{\tiny{3}};
\node at (7.9+5,1.2) []{\tiny{3}};
\node at (6.4+5,1.7) []{\tiny{3}};
\node at (6.9+5,0.7) []{\tiny{3}};
\node at (6.9+5,1.7) []{\tiny{3}};
\node at (6.9+5,1.2) []{\tiny{3}};
\node at (7.4+5,0.7) []{\tiny{3}};
\node at (7.4+5,1.7) []{\tiny{3}};
\node at (7.4+5,1.2) []{\tiny{3}};

\node at (6.5,-0.2) {(a)};
\node at (6.5+5,-0.2) {(b)};
\end{tikzpicture}
\end{center}

\begin{center}
\begin{tikzpicture}[scale=1]

\foreach \x in {2,2.5,3,...,6}
{
\draw (\x+2.5,0) -- (\x+2.5,2);
}
\foreach \y in {0,0.5, ...,2}
{
\draw (2+2.5,\y) -- (6+2.5,\y);
}

\draw[ultra thick, color=red] (5.5,1) -- (6,1);
\draw[ultra thick, color=red] (5.5,1.5) -- (7.5,1.5);
\draw[ultra thick, color=red] (6,1.5) -- (6,1);
\draw[ultra thick, color=red] (7.5,1.5) -- (7.5,0.5);
\draw[ultra thick, color=red] (7,1) -- (8,1);
\draw[ultra thick, color=red] (7,0.5) -- (8,0.5);
\draw[ultra thick, color=red] (7,0.5) -- (7,0);

\node at (6,1) [circle,draw=black,fill=black,scale=0.5]{};
\node at (5.5,1) [circle,draw=black,fill=black,scale=0.5]{};
\node at (5.5,1.5) [circle,draw=black,fill=black,scale=0.5]{};
\node at (6,1.5) [circle,draw=black,fill=black,scale=0.5]{};
\node at (8,0.5) [circle,draw=black,fill=black,scale=0.5]{};
\node at (7,0) [circle,draw=black,fill=black,scale=0.5]{};
\node at (7,0.5) [circle,draw=black,fill=black,scale=0.5]{};
\node at (8,1) [circle,draw=black,fill=black,scale=0.5]{};
\node at (6.5,1.5) [circle,draw=black,fill=black,scale=0.5]{};
\node at (7,0.5) [circle,draw=black,fill=black,scale=0.5]{};
\node at (7,1.5) [circle,draw=black,fill=black,scale=0.5]{};
\node at (7,1) [circle,draw=black,fill=black,scale=0.5]{};
\node at (7.5,0.5) [circle,draw=black,fill=black,scale=0.5]{};
\node at (7.5,1.5) [circle,draw=black,fill=black,scale=0.5]{};
\node at (7.5,1) [regular polygon, regular polygon sides=4,draw=black,fill=red,scale=0.5]{};

\node at (6.2,1.2) []{3};
\node at (5.7,1.2) []{4};
\node at (5.7,1.7) []{9};
\node at (6.2,1.7) []{8};
\node at (8.2,0.7) []{7};
\node at (7.2,0.2) []{4};
\node at (8.2,1.2) []{12};
\node at (6.7,1.7) []{7};
\node at (7.2,0.7) []{9};
\node at (7.2,1.7) []{6};
\node at (7.2,1.2) []{1};
\node at (7.7,0.7) []{8};
\node at (7.7,1.7) []{5};
\node at (7.7,1.2) []{0};

\node at (6.5,-0.2) {(c)};

\end{tikzpicture}
\end{center}
\caption{One spanning tree of particles, a possible numbering of the ports of the particles before (Figure~\ref{portreconfig}.a) and after the execution of Algorithm~\ref{alg3} (Figure~\ref{portreconfig}.b) and the 4-identifier obtained by executing Algorithm~\ref{alg4} (Figure~\ref{portreconfig}.c) in the square grid (square: leader; thick line: edge of the spanning tree; small number: port number of a particle; big number: 4-identifier of a particle).}
\label{portreconfig}
\end{figure}

The idea behind Algorithm~\ref{alg3} is to reproduce, in each particle, the way the ports are numbered in the leader particle. To achieve this goal, each particle $p$ receives a message from its parent containing the port number of the parent connected to $p$ and $p$ renumbers its own ports in order that its port numbers are coherent with the sent number. Figure~\ref{portreconfig}.a and Figure~\ref{portreconfig}.b illustrate the port numbers of particles before and after the execution of Algorithm~\ref{alg3}.

\subsection{The $k$-local identifiers}

Now, we aim to give to each particle a variable $id$, called its $k$-local identifier, such that every two particles $p_{1}$ and $p_{2}$ with the same identifier satisfy $d_{G}(s(p_{1}),s(p_{2}))>k$. If we suppose that the particles have not a memory of at least $\log_{2}(n)$ bits, for $n=|P|$, then it is not possible to record a unique variable for each particle.
However, it is possible to have a $k$-local identifier in the three considered grids only using at most $\log_{2}((k+1)^{2})$ bits where $k$ is a parameter given by the user. 
Our proposed Algorithm~\ref{alg4} presents an optimal way (in term of memory) to compute $k$-local identifiers. We suppose that the port renumbering algorithm (Algorithm~\ref{alg3}) has been done before executing Algorithm~\ref{alg4}.

\begin{algorithm}
\caption{The $k$-local identifier algorithm for a particle $p$.} 
\label{alg4}
\begin{algorithmic} 
\State \textbf{Case 1:} State \textbf{L}.
\State set $i=0$, $j=0$, $id=0$;
\State send $i$ and $j$ through each port from $\text{child}(p)$.
\State \textbf{Case 2:} State \textbf{N}.
\If { the particle receives the integers $i'$ and $j'$ through the port $a$}
    \State set $i=I^{k}_{G}(i',a)$, $j=I^{k}_{G}(j',a)$, $id=f^k_{G}(i,j)$;
    \State send $i$ and $j$ through each port from $\text{child}(p)$.
\EndIf
\end{algorithmic}
\end{algorithm}

Algorithm~\ref{alg4} consists in assigning a variable which corresponds to a color in a coloring of the $k^{\text{th}}$  power on the grid. 
More precisely, the function $f_G^k$ consists in assigning a color depending the Cartesian coordinate of the vertices. Since the colors are given following a pattern, the Cartesian coordinate can be stored relatively to the size of the patterns.
In Algorithm~\ref{alg4}, the leader affects to itself the color $0$ and following the direction where the messages are transmitted, the particles reproduce the coloring patterns given in Appendix \ref{appendixd}. 
The functions $f^k_G$ and $I^k_G$, $J^k_G$, used in Algorithm~\ref{alg4} are defined, depending on the choice of $G$, as follows: $f^k_{\mathscr{S}}(i,j)= (i+kj) \pmod{m_{k} }$, $f^k_{\mathscr{K}}(i,j)= i_{\pmod{k+1}} +(k+1)j_{\pmod{k+1}}$ and $ f^k_{\mathscr{T}}(i,j)= (i_{\pmod{3(k+1)/2}}+j(3(k+1)/2)+ \lfloor 2 j/(k+1) \rfloor (k+1)/2)) \pmod{m'_{k}}$ if $k$ is odd or $f^k_{\mathscr{T}}(i,j)= (i+(3k/2+1) j) \pmod{m'_{k}}$ otherwise; $I_{G}^{k}(i,a)=i$ if $(a=1;3\land G\cong \mathscr{S})\lor (a=1;4\land G\cong \mathscr{T})\lor (a=2;6\land G\cong \mathscr{K})$, $I_{\mathscr{S}}^{k}(i,a)= i+1 \pmod{\lceil (k+1)^2/2 \rceil}$ if $a=0$, $I_{\mathscr{S}}^{k}(i,a)= i-1 \pmod{\lceil (k+1)^2/2 \rceil}$ if $a=2$, $I_{\mathscr{T}}^{k}(i,a)= i+1 \pmod{\lceil 3(k+1)^2/4 \rceil}$ if $a=0;5$, $I_{\mathscr{T}}^{k}(i,a)= i-1 \pmod{\lceil 3(k+1)^2/4 \rceil}$ if $a=2;3$, $I_{\mathscr{K}}^{k}(i,a)= i+1 \pmod{k+1}$ if $a=0;1;7$ and $I_{\mathscr{K}}^{k}(i,a)= i-1 \pmod{k+1}$ if $a=3;4;5$; $J_{G}^{k}(j,a)=i$ if $(a=0;2\land G\cong \mathscr{S})\lor (a=0;6\land G\cong \mathscr{T}) \lor (a=0;4 \land G\cong \mathscr{K})$, $J_{\mathscr{S}}^{k}(j,a)= j+1 \pmod{\lceil (k+1)^2/2 \rceil}$ if $a=1$, $J_{\mathscr{S}}^{k}(j,a)= i-1 \pmod{\lceil (k+1)^2/2 \rceil}$ if $a=3$, $J_{\mathscr{T}}^{k}(j,a)= i+1 \pmod{\lceil 3(k+1)^2/4 \rceil}$ if $a=1;2$, $J_{\mathscr{T}}^{k}(j,a)= i-1 \pmod{\lceil 3(k+1)^2/4 \rceil}$ if $a=4;5$, $J_{\mathscr{K}}^{k}(j,a)= i+1 \pmod{k+1}$ if $a=1;2;3$ and $J_{\mathscr{K}}^{k}(j,a)= i-1 \pmod{k+1}$ if $a=5;6;7$. Note that the functions $I^k_G$ and $J^k_G$ are used to determine the Cartesian coordinate of a particle using the Cartesian coordinate of a neighbor and the port number of this neighbor.

Since the values of $f_{G}^{k}(i,j)$ is bounded by $3(k+1)^2/4$, if $G$ is isomorphic to one of the three grids, the size of the messages will not exceed $\log_2(3(k+1)^2/4)$. As for Algorithm~\ref{alg3}, the number of sent messages is $|V(P)|-1$. Figure~\ref{portreconfig}.c illustrates the obtained 4-identifiers after the execution of Algorithm~\ref{alg4}.

Since the particles can move during the execution of an algorithm, the $k$-local identifiers may become not valid anymore ( i.e., there may be two particles $p_1$ and $p_2$ with the same $k$-local identifier and with $d_G(s(p_1),s(p_2))\le k$) if the structure of the particle graph $P$ on $G$ changes.
It is possible to keep a valid $k$-local identifier in case a particle moves in a direction of a port $a$ by setting $id=f^k_{G}(I^{k}_{G}(i,r_{G}(a)),J^{k}_{G}(j,r_{G}(a)))$ as the new $k$-local identifier. It corresponds to update the variable $id$ which corresponds to a color in a coloring of the $k^{\text{th}}$ power on the grid in function of the new position of the particle.
Also, in the case a particle do $\ell$ movements, by storing the successive directions of movement of the particle during these $\ell$ movements, it is also possible to update the value of the $k$-local identifier in order that it remains valid.

Note that for both Algorithms~\ref{alg3} and~\ref{alg4} finish after at most $h$ rounds, $h$ being the height of the spanning tree. Also the number of sent messages in both Algorithms~\ref{alg3} and~\ref{alg4} is $|V(G)|-1$ (the number of edges in a spanning tree).

\section{Conclusion}
In this paper, we have presented a new leader election algorithm based on local computation. We have also presented an algorithm which affects a different variable for every two particles $p_{1}$ and $p_{2}$ at distance at most $k$. All the presented algorithms only require a $O(1)$-space memory. This complexity makes it possible to use our algorithms for programmable matter. Moreover, in case of movements of particles, there is no need of communication in order to update the $k$-local identifiers.

As future work, it would be interesting to determine a more general deterministic leader election algorithm in our algorithmic framework that can take into account fault tolerance. Also, it would be interesting to extend the presented results to 3D grids. Another interesting question could be to use our results to clustering the set of particles in several sets which induce subgraphs of small diameter.
\section*{Acknowledgments}
This work was supported by the French "Investissements d'Avenir" program, project ISITE-BFC (contract ANR-15-IDEX-03).

\begin{subappendices}
\renewcommand{\thesection}{\Alph{section}}%
\section{Proof of Theorem \ref{theobavure} and bound on the complexity of the $S$-contraction algorithm \label{appendixa}}
The three lemmas presented in this appendix are used in order to prove Theorem \ref{theobavure}.

In the following lemma, we describe how to determine, in the context of programmable matter, if a particles is $S$-contractible or not.
\begin{lem}\label{artornot}
Let $G$ be an infinite grid among $\mathscr{S}$, $\mathscr{T}$ and $\mathscr{K}$ and let $S\subseteq V(P)$, for $P$ the particle graph on $G$ and $S$ the set of vertices occupied by all the particles in the same fixed state.
One round is sufficient in order that every particle determine if it is $S$-contractible or not if $G$ is isomorphic to $\mathscr{S}$.
Otherwise, if $G$ is isomorphic to $\mathscr{T}$ or $\mathscr{K}$, no round is necessary.
\end{lem}

\begin{proof}
Let $N^{+}(p)$ be the set of port labels on which $p$ can communicate with particles from its neighborhood.
In order to verify that $M_{G}(p)\cap S$ is connected in the triangular or king grids, it suffices to verify that $N^{+}_{G}(p)\cap S$ forms an interval of consecutive integers (by considering that 0 and $\Delta(G)-1$ are consecutive).
For example, $\{0 ,4,5\}$ contains successive integers in the triangular grid but that is not the case for $\{0,2,5\}$. Such verification in the triangular and king grids can be done during any local computation.
Figure \ref{figstmac} illustrates three possible cases that could happen for a particle in the triangular grid. On the left part of Figure \ref{figstmac}, the particle does not satisfy Property I) but satisfies Property II). On the middle part of Figure \ref{figstmac}, the particle satisfies Property I) and does not satisfy Property II). Finally, on the right part of Figure \ref{figstmac}, the particle satisfies both Properties I) and II).

In the square grid, in order to test if a particle $p$ is such that $M_{G}(p)\cap S$ is connected, it requires to receive $N^{+}_{G}(p')\cap S$, from the particle $p'$ in the neighborhood of $p$ and afterward to test if $N^{+}_{G}(p)$ only contains consecutive integers (by considering that 0 and 3 are consecutive) and then to verify, for any two successive particles $p'$ and $p''$ from the neighborhood, that the vertex which corresponds to the corner adjacent to both $p'$ and $p''$ is occupied by a particle.

If $G$ is among $\mathscr{T}$ and $\mathscr{K}$, then no round is required to know if a particle is in $S$ or not (since a particle know the state of its neighbors). If $G$ is isomorphic to $\mathscr{S}$, then, in one round, which consists in sending the values of $N^{+}(p)\cap S$ to the adjacent particles, every particle knows if it is $S$-contractible or not.
$\hfill\qed$
\end{proof}
The following lemma is be useful in order to prove that our leader election algorithm works correctly.

\begin{lem}\label{lem11}
Let $G$ be an infinite grid among $\mathscr{S}$, $\mathscr{T}$ and $\mathscr{K}$ and let $S\subseteq V(P)$, for $P$ the particle graph on $G$. Let $p$ be an $S$-contractible particle.
If $S$ is connected and hole-free, then $S-\{s(p)\}$ is connected and hole-free.
\end{lem}
\begin{proof}
First, note that in all three grids, the fact that $|N_{G}(p)\cap S|<|N_{G}(p)|$ implies that there is a vertex $v$ in $N_{G}(p)\setminus S$. By contradiction, suppose we create a hole in $G[S]$ by removing the vertex $s(p)$ from $S$. This implies, since $G[M_{G}(p)\cap S]$ is connected, that $v$ was already in a hole from $G[S]$.
Second, since $G[M_{G}(p)\cap S]$ is connected, we are sure that the subgraph $G[S\setminus\{s(p)\}]$ is connected.
$\hfill\qed$
\end{proof}

To ensure that our leader election algorithm works correctly, it remains to prove that there always exists an $S$-contractible particle. That is what we do in the following Lemma.
\begin{lem}\label{propexone}
Let $S\subseteq V(P)$, for $P$ the particle graph on $G$. If $G[S]$ is hole-free, then there always exists an $S$-contractible particle in $S$.
\end{lem}
\begin{proof}
Note that there exists a particle on the border of $G[S]$ since $S$ is finite. Let $A$ be the set of particles on the border of $G[S]$. 
For any particle $p$, the fact that there are at least two connected components $B_1$ and $B_2$ in $G[M_G(p)\cap S]$ implies that there is no path in $G[S\setminus \{s(p)\}]$ between any vertex of $B_1$ and a vertex of $B_2$, since it would imply the existence of a hole in $G[S]$ containing a vertex from $M_G(p)\setminus S$.
Therefore, if $p\in A$ and if $p$ is not an $S$-contractible particle, then $p$ is an articulation of $G[S]$.

Now suppose, by contradiction, that there is no $S$-contractible particle in $S$. By the previous remark, the graph $G[A]$ is connected and all particles of $A$ are articulations of $G[S]$. However, a finite graph containing a cycle contains vertices which are not articulation of $G[S]$. Thus, $G[A]$ contains no cycle ($G[A]$ is a forest). However, by definition, the leaves (the vertices of degree $1$) are $S$-contractible. Thus, we obtain a contradiction with the fact that there is no $S$-contractible particle in $S$.
$\hfill\qed$
\end{proof}
In the case $P$ is hole-free, note that by Lemma~\ref{lem11} and Lemma~\ref{propexone} there is always a particle which is both on the border of $G[S]$ and not an articulation of $G[S]$.

\begin{proof}[Proof of Theorem \ref{theobavure}]
Note that before the execution of the algorithm, the set $S$ is the set $V(P)$. Since $P$ is hole-free and connected and by Lemma~\ref{lem11}, $S$ remains connected and hole-free during the execution of the algorithm.
By Lemma~\ref{propexone}, there is always a particle in $S$ which is $S$-contractible (every particle on the border which is not an articulation is $S$-contractible). Thus, for every round, the number of particles in state \textbf{C} strictly decreases. Since $|V(P)|$ is finite, we are sure that at some point, $S$ will only contain one vertex. If at some point, $S$ contains one vertex, there will be at least one elected leader.

Finally, note that the fact that there are two elected leaders contradicts the fact that $S$ remains connected during the execution of the algorithm.
$\hfill\qed$
\end{proof}

Let $G'$ be a subgraph of $G$ such that $G'$ is hole-free, the \textit{radius} of $G'$, denoted by $r(G')$, is given by $r(G')=\min_{u\in V(G')}$ $\{ \max_{v\in A} ( d_{G'}(u,v)) \}$, for $A$ the set of the particles on the border of $G'$.
Moreover, let $h(T)$ be the \textit{height} of a tree $T$, i.e., $h(T)=$ $\min_{u\in V(T)}$ $\{\max_{v\in V(T),\ |N_{T}(v)|=1}$ $ (d_{T}(u,v))\}$ and let $mtree(G')$ be the maximum height among all induced subgraphs of $G'$ which are trees, i.e., among the set $\{ G'[B] |\ B\subseteq V(G'),\ G'[B]$ is a tree $\}$.

In the following Proposition, we give a bound on the required number of rounds for the termination of Algorithm \ref{alg1}.

\begin{prop}
Let $S$ be the set of particles in state \textbf{C} and $P$ be the particle graph on $G$.
Moreover, let $b_G=r(P)+mtree(P)+1$ if $G$ is isomorphic to $\mathscr{T}$ or $\mathscr{K}$ or $b_G=2(r(P)+mtree(P))+2$ if $G$ is isomorphic to $\mathscr{S}$.
If $P$ is hole-free, then after $b_G(P)$ rounds of the $S$-contraction algorithm on $P$, one particle will be the leader.

\end{prop}
\begin{proof}
First, suppose $G$ is isomorphic to $\mathscr{T}$ or $\mathscr{K}$.
Let $S_t$ be the set of particles in state \textbf{C} after the first $t$ rounds.
Note that after $r(P)+1$ rounds we are sure that every remaining particle $u$ satisfies $|N_{G}(u)\cap S_{r(S)+1}|<N_{G}(u)$. 
This is due to the fact that each particle $u$ on the border of $S_i$, for $i\ge 0$, is not in $S_{i+1}$ if $|N_{G}(u)\cap S_{r(S)+1}|<N_{G}(u)$.
Thus, by Lemma~\ref{lem11}, $G[S_{r(P)+1}]$ is either a tree or empty.
By definition, we have $h(G[S_{r(P)+1}])\le mtree(P)$. Note that in the case $G[S_{t}]$ is not a trivial tree (a tree containing only one vertex), we have $h(G[S_{t+1}])=h(G[S_{t}])-1$, for $t\ge r(P)+1$. Therefore, we obtain that $S_{t}$ is empty if $t\ge r(P)+mtree(P)+1$.

Second, suppose $G$ is isomorphic to $\mathscr{S}$.
Note that, by Lemma~\ref{artornot}, one round is sufficient in order that every particle determines if it is $S$-contractible.
Consequently, it is easy to observe that the required number of rounds in order that the $S$-contraction algorithm finishes for $\mathscr{S}$ is bounded by two times the required number of rounds in order that the $S$-contraction algorithm finishes for $\mathscr{T}$ or $\mathscr{K}$. 
$\hfill\qed$
\end{proof}

\section{An example of algorithm in order to construct a spanning tree\label{appendixb}}

Our proposed algorithm (Algorithm~\ref{alg2}) for constructing a spanning tree consists in setting the particle in state \textbf{L} as the root and, afterward, constructing a spanning tree using a classical distributed spanning tree algorithm.

\begin{algorithm}
\caption{A spanning-tree algorithm for a particle $p$.} 
\label{alg2}
\begin{algorithmic} 
\State \textbf{Case 1:} State \textbf{L} (leader).
\State set $\text{child}(p)= N^{+}_{G}(p)$; 
\State send a message $m$ (which only contains the bit $0$) through each port from $\text{child}(p)$.
\State \textbf{Case 2:} State \textbf{N} (not elected).
\If { the particle receives the message $m$ through the port $a$}
    \If{ the particle has never received the message $m$ before}
        \State set $\text{parent}(p)=a$;
        \State set $\text{child}(p)=N_{G}^{+}(p)\setminus\{a_{1},\ldots, a_{\ell}\}$, where $a_1$, $\ldots$, $a_{\ell}$ are ports on which $p$ has received the message $m$;
        \State send the message $m$ through each port from $\text{child}(p)$.
    \EndIf
\EndIf
\end{algorithmic}
\end{algorithm}

\section{Combining the $S$-contraction algorithm with a general leader election algorithm\label{appendixc}}

Daymude et al. \cite{JJD2017} introduced a leader election algorithm that works on every configuration of $P$. In this appendix we present a way to reduce the required number of rounds in order that this algorithm finishes its execution (by using the $S$-contraction algorithm).
In the remaining part of this appendix, the leader election algorithm from \cite{JJD2017} will be called the \emph{general leader election algorithm} (to the best of our knowledge, it is the only leader election algorithm for programmable matter working on every configuration).

In order to simplify the presentation of the results, we only discuss the results for the triangular grid. However, by modifying the algorithms it is possible to make it work for the square and king grids also. We begin this appendix by describing how the general leader election algorithm works. This description will help the reader to convince itself that, in a lot of cases, combining the $S$-contraction algorithm with a general leader election algorithm could be a good idea.

For particle $p$ and a port $a$ of $p$ connected to another particle, we denote by $n(a)$ the port number of the first port of $p$ connected to a particle after $a$ in the clockwise order.
The general leader election algorithm uses the fact that it is possible to send a message around a boundary.
Sending a message around a boundary consists in sending and re-transmitting the message in the following way: for a particle $p$, if a message is received from the port $a$, then the particle re-transmits the message to the particle connected to $p$ by the port $n(a)$ of $p$. Figure~\ref{artfig} represents how the messages are transmitted in this case.

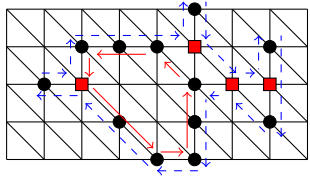
\begin{figure}[t]
\begin{center}
\begin{tikzpicture}[scale=1]

\foreach \x in {2,2.5,3,...,6}
{
\draw (\x+2.5,0) -- (\x+2.5,2);
}
\foreach \y in {0,0.5, ...,2}
{
\draw (2+2.5,\y) -- (6+2.5,\y);
}
\foreach \z in {2,2.5,3,3.5,4}
{
\draw (\z+2.5,2) -- (\z+2.5+2,0);
}
\foreach \h in {0.5,1,1.5}
{
\draw (2+2.5,\h) -- (\h+2.5+2,0);
\draw (\h+4+2.5,2) -- (6+2.5,\h);
}

\node at (5,1) [circle,draw=black,fill=black,scale=0.5]{};
\node at (6,0.5) [circle,draw=black,fill=black,scale=0.5]{};
\node at (5.5,1) [regular polygon, regular polygon sides=4,draw=black,fill=red,scale=0.5]{};
\node at (5.5,1.5) [circle,draw=black,fill=black,scale=0.5]{};
\node at (6,1.5) [circle,draw=black,fill=black,scale=0.5]{};
\node at (6.5,0) [circle,draw=black,fill=black,scale=0.5]{};
\node at (7,0) [circle,draw=black,fill=black,scale=0.5]{};
\node at (7,0.5) [circle,draw=black,fill=black,scale=0.5]{};
\node at (7,1) [circle,draw=black,fill=black,scale=0.5]{};
\node at (6.5,1.5) [circle,draw=black,fill=black,scale=0.5]{};
\node at (7,0.5) [circle,draw=black,fill=black,scale=0.5]{};
\node at (7,1.5) [regular polygon, regular polygon sides=4,draw=black,fill=red,scale=0.5]{};
\node at (7,2) [circle,draw=black,fill=black,scale=0.5]{};
\node at (7.5,1) [regular polygon, regular polygon sides=4,draw=black,fill=red,scale=0.5]{};
\node at (8,1.5) [circle,draw=black,fill=black,scale=0.5]{};
\node at (8,1) [regular polygon, regular polygon sides=4,draw=black,fill=red,scale=0.5]{};
\node at (8,0.5) [circle,draw=black,fill=black,scale=0.5]{};

\draw[<-,color=red] (5.6,1.1) -- (5.6,1.35);
\draw[<-,color=red] (5.7,1.4) -- (6.35,1.4);
\draw[<-,color=red] (6.6,1.3) -- (6.8,1.1);
\draw[<-,color=red] (6.9,0.9) -- (6.9,0.15);
\draw[<-,color=red] (6.85,0.1) -- (6.55,0.1);
\draw[<-,color=red] (6.45,0.15) -- (5.65,0.95);

\draw[<-,dashed,color=blue]  (5.55,0.75)-- (6.4,-0.1);
\draw[<-,dashed,color=blue]  (6.5,-0.15) -- (7.05,-0.15);
\draw[<-,dashed,color=blue]  (7.15,-0.1) -- (7.15,0.8);
\draw[<-,dashed,color=blue]  (7.2,0.85) -- (7.4,0.85);
\draw[<-,dashed,color=blue]  (7.55,0.75) -- (7.95,0.35);
\draw[<-,dashed,color=blue]  (8.15,0.3) -- (8.15,1.65);
\draw[<-,dashed,color=blue]  (7.85,1.65) -- (7.85,1.2);
\draw[<-,dashed,color=blue]  (7.8,1.15) -- (7.6,1.15);
\draw[<-,dashed,color=blue]  (7.55,1.15) -- (7.15,1.55);
\draw[<-,dashed,color=blue]  (7.15,1.6) -- (7.15,2.1);
\draw[<-,dashed,color=blue]  (6.85,2.1) -- (6.85,1.7);
\draw[<-,dashed,color=blue]  (6.8,1.65) -- (5.4,1.65);
\draw[<-,dashed,color=blue]  (5.35,1.6) -- (5.35,1.2);
\draw[<-,dashed,color=blue]  (5.3,1.15) -- (4.9,1.15);
\draw[<-,dashed,color=blue]  (4.9,0.85) -- (5.45,0.85);

\end{tikzpicture}
\end{center}
\caption{The way how the messages are re-transmitted in the algorithm of Daymude et al. \cite{JJD2017} (square: articulation; dashed arrow: message transmitted along the border; simple arrow: message transmitted along the hole).}
\label{artfig}
\end{figure}

For each hole $H$ of the particle graph $P$ on $G$, we denote by $b(H)$, the set of particles which are adjacent with a vertex of $H$. Also, when a particle is adjacent to vertices of different holes or when a particle belongs to the border of $P$, it is possible to decompose particles in agents, each agent corresponding to a different hole or to the border. Thus, an agent will be either adjacent to vertices of at most one hole or belong to the border but never both.

The general leader election algorithm can be summarized as the succession of four phases. A first phase consists in removing the candidacy of each particle having six neighbors. A second phase consists, for each hole $H$ of $P$, to remove the candidacy of some agents of $b(H)$ using a randomized procedure. Simultaneously, the same process is done for agents in the border of $P$. A third phase consists in verifying if there is only one candidate in $b(H)$, for each hole $H$ of $P$ and only one candidate in the border of $P$. They calculate the relative positions of the candidates in order to do such verification. Finally, in the last phase, they verify if the remaining candidates agents are in $b(H)$ or in border of $P$. The leader will be the candidate particle of the border of $P$.
We can verify if an agent is in $b(H)$ by sending a message around the boundary and verifying when the message comes back to the initial particle if this message has been re-transmitted in the clockwise direction or not. 
As Figure~\ref{artfig} illustrates, the messages re-transmitted through the boundary of a hole are re-transmitted in the counterclockwise direction and the messages re-transmitted through the border of $P$ are re-transmitted in the clockwise direction.
In all these phases, the messages are re-transmitted around a boundary.

The required number of rounds in order that the general leader election algorithm finishes its execution is $O(\ell)$, where $\ell$ is the number of particles in the border of $P$.

We do the following remark about the $S$-contraction algorithm that comes from the fact that for any $S$-contractible particle $p$ of $S$, if $G[S]$ is connected, then $G[S\setminus \{s(p)\}]$ is also connected.

\begin{re}
If the particle graph $P$ has a hole then, after the execution of the $S$-contraction algorithm on the particle graph $P$ on $G$ there will remain particles in state \textbf{C}. Also, the graph induced by the particles in state \textbf{C} is connected.
\end{re}
In particular when $P$ has one hole, the remaining particles in state \textbf{C} will form a ring in the triangular grid.
Thus, it is possible to run the $S$-contraction algorithm and, afterward, execute the general leader election algorithm on the remaining particle in state \textbf{C}. Let $S_c$ be the set of particles in state \textbf{C} after the $S$-contraction algorithm on the particle graph $P$ on $G$. Depending on the structure of $P$, it could happen that the number of particles in the border $\mathcal{T}[S_c]$ is smaller than in the border of $P$ and that it speeds the execution of the general leader election algorithm. For example, that is always the case when $P$ has at most one hole.

\section{Coloring the $k^{\text{th}}$ power of graph\label{appendixd}}
The $k^{\text{th}}$ power of a graph $G$ is the graph on the same vertex set than $G$ and with edges connecting every two vertices $u$ and $v$ satisfying $d_{G}(u,v)\le k$. Note that there is a correlation between this definition and the $k^{\text{th}}$ power of the adjacency matrix of $G$ (the adjacency matrix of the  $k^{\text{th}}$ power of $G$ is easily obtained from this matrix).
Our goal in this appendix is to determine an optimal coloring of the $k^{\text{th}}$  power of the square, triangular and king grids.
We use these colorings in order to propose a distributed algorithm (supposing we have a leader) in order to assign $k$-local identifiers to the particles (see Section 4).
A coloring of the $k^{\text{th}}$  power of a grid corresponds to assign a value to each vertex of the graph such that every two vertices with the same assigned value are at distance at least $k+1$.
An example of coloring of the $k^{\text{th}}$ power of the square grid is represented by Figures \ref{motifdistance3} and \ref{motifdistance4}.
In Figure \ref{motifdistance3}, it is easy to notice that every two vertices with color 0 (or any other color) are at distance at least $4$.

More formally, a \emph{$k$-coloring} of a graph $G$ is a map $c$ from $V(G)$ to $\{0,1,\ldots,k-1\}$ which satisfies $c(u)\neq c(v)$ for every $uv\in E(G)$.
The \emph{chromatic number} $\chi(G)$ of $G$, is the smallest integer $k$ such that there exists a $k$-coloring of $G$.
The \emph{$k^{\text{th}}$  power} $G^k$ of a graph $G$ is the graph obtained from $G$ by adding an edge between every two vertices satisfying $d_{G}(u,v)\le k$. 
More details about the coloring of the $k^{\text{th}}$  power of graphs can be found in the survey from Kramer and Kramer~\cite{KR2008}. The results presented in this appendix are inspired by the previous works \cite{FE2003,JJ2005,SE2001} about the coloring of the $k^{\text{th}}$  power of the grids.

\subsection{Coloring the $k^{\text{th}}$  power of square grids}
We give the following result from Fertin et al.~\cite{FE2003}.
\begin{theo}[\cite{FE2003}]
For any $k\ge 1$, $\chi(\mathscr{S}^{k})= \lceil (k+1)^2/2 \rceil$.
\end{theo}
Let $m_{k}=\lceil (k+1)^2/2 \rceil$.
In their paper, Fertin et al. define an optimal coloring $c$ of the $k^{\text{th}}$  power of the square grid as follows: $c((i,j))= (i+kj) \pmod{m_{k}}$.
In Figures \ref{motifdistance3} and \ref{motifdistance4}, we represent patterns for coloring the $3th$ and $4th$ powers of the square grid. These patterns have been obtained using the coloring from~\cite{FE2003}.
Note that since there is a pattern, a vertex $(i,j)$ can determine its color only knowing $i\pmod{m_{k}}$ and $j\pmod{m_{k}}$.
We recall the definition of the following function $f^k_{\mathscr{S}}(i,j)= (i+kj) \pmod{m_{k} }$. Note that $f^k_{\mathscr{S}}(i,j)=f^k_{\mathscr{S}}(i',j')$, in the case $i\equiv i'\pmod{m_{k}}$ and  $j\equiv j'\pmod{m_{k}}$
This function is used in order to assign $k$-local identifiers to particles.
\begin{figure}
\begin{center}
\begin{tikzpicture}[scale=0.7]
\node at (0,3) {$0$};
\node at (1,3) {$1$};
\node at (2,3) {$2$};
\node at (3,3) {$3$};
\node at (4,3) {$4$};
\node at (5,3) {$5$};
\node at (6,3) {$6$};
\node at (7,3) {$7$};
\node at (0,2) {$3$};
\node at (1,2) {$4$};
\node at (2,2) {$5$};
\node at (3,2) {$6$};
\node at (4,2) {$7$};
\node at (5,2) {$0$};
\node at (6,2) {$1$};
\node at (7,2) {$2$};
\node at (0,1) {$6$};
\node at (1,1) {$7$};
\node at (2,1) {$0$};
\node at (3,1) {$1$};
\node at (4,1) {$2$};
\node at (5,1) {$3$};
\node at (6,1) {$4$};
\node at (7,1) {$5$};
\node at (0,0) {$1$};
\node at (1,0) {$2$};
\node at (2,0) {$3$};
\node at (3,0) {$4$};
\node at (4,0) {$5$};
\node at (5,0) {$6$};
\node at (6,0) {$7$};
\node at (7,0) {$0$};
\node at (0,-1) {$4$};
\node at (1,-1) {$5$};
\node at (2,-1) {$6$};
\node at (3,-1) {$7$};
\node at (4,-1) {$0$};
\node at (5,-1) {$1$};
\node at (6,-1) {$2$};
\node at (7,-1) {$3$};
\node at (0,-2) {$7$};
\node at (1,-2) {$0$};
\node at (2,-2) {$1$};
\node at (3,-2) {$2$};
\node at (4,-2) {$3$};
\node at (5,-2) {$4$};
\node at (6,-2) {$5$};
\node at (7,-2) {$6$};
\node at (0,-3) {$2$};
\node at (1,-3) {$3$};
\node at (2,-3) {$4$};
\node at (3,-3) {$5$};
\node at (4,-3) {$6$};
\node at (5,-3) {$7$};
\node at (6,-3) {$0$};
\node at (7,-3) {$1$};
\node at (0,-4) {$5$};
\node at (1,-4) {$6$};
\node at (2,-4) {$7$};
\node at (3,-4) {$0$};
\node at (4,-4) {$1$};
\node at (5,-4) {$2$};
\node at (6,-4) {$3$};
\node at (7,-4) {$4$};
\end{tikzpicture}
\end{center}
\caption{A pattern for coloring the $3^{\text{th}}$ power of the square grid.}
\label{motifdistance3}
\end{figure}
\begin{figure}
\begin{center}
\begin{tikzpicture}[scale=0.65]
\node at (0,12) {$0$};
\node at (1,12) {$1$};
\node at (2,12) {$2$};
\node at (3,12) {$3$};
\node at (4,12) {$4$};
\node at (5,12) {$5$};
\node at (6,12) {$6$};
\node at (7,12) {$7$};
\node at (8,12) {$8$};
\node at (9,12) {$9$};
\node at (10,12) {$10$};
\node at (11,12) {$11$};
\node at (12,12) {$12$};
\node at (0,11) {$5$};
\node at (1,11) {$6$};
\node at (2,11) {$7$};
\node at (3,11) {$8$};
\node at (4,11) {$9$};
\node at (5,11) {$10$};
\node at (6,11) {$11$};
\node at (7,11) {$12$};
\node at (8,11) {$0$};
\node at (9,11) {$1$};
\node at (10,11) {$2$};
\node at (11,11) {$3$};
\node at (12,11) {$4$};
\node at (0,10) {$10$};
\node at (1,10) {$11$};
\node at (2,10) {$12$};
\node at (3,10) {$0$};
\node at (4,10) {$1$};
\node at (5,10) {$2$};
\node at (6,10) {$3$};
\node at (7,10) {$4$};
\node at (8,10) {$5$};
\node at (9,10) {$6$};
\node at (10,10) {$7$};
\node at (11,10) {$8$};
\node at (12,10) {$9$};
\node at (0,9) {$2$};
\node at (1,9) {$3$};
\node at (2,9) {$4$};
\node at (3,9) {$5$};
\node at (4,9) {$6$};
\node at (5,9) {$7$};
\node at (6,9) {$8$};
\node at (7,9) {$9$};
\node at (8,9) {$10$};
\node at (9,9) {$11$};
\node at (10,9) {$12$};
\node at (11,9) {$0$};
\node at (12,9) {$1$};
\node at (0,8) {$7$};
\node at (1,8) {$8$};
\node at (2,8) {$9$};
\node at (3,8) {$10$};
\node at (4,8) {$11$};
\node at (5,8) {$12$};
\node at (6,8) {$0$};
\node at (7,8) {$1$};
\node at (8,8) {$2$};
\node at (9,8) {$3$};
\node at (10,8) {$4$};
\node at (11,8) {$5$};
\node at (12,8) {$6$};
\node at (0,7) {$12$};
\node at (1,7) {$0$};
\node at (2,7) {$1$};
\node at (3,7) {$2$};
\node at (4,7) {$3$};
\node at (5,7) {$4$};
\node at (6,7) {$5$};
\node at (7,7) {$6$};
\node at (8,7) {$7$};
\node at (9,7) {$8$};
\node at (10,7) {$9$};
\node at (11,7) {$10$};
\node at (12,7) {$11$};
\node at (0,6) {$4$};
\node at (1,6) {$5$};
\node at (2,6) {$6$};
\node at (3,6) {$7$};
\node at (4,6) {$8$};
\node at (5,6) {$9$};
\node at (6,6) {$10$};
\node at (7,6) {$11$};
\node at (8,6) {$12$};
\node at (9,6) {$0$};
\node at (10,6) {$1$};
\node at (11,6) {$2$};
\node at (12,6) {$3$};
\node at (0,5) {$9$};
\node at (1,5) {$10$};
\node at (2,5) {$11$};
\node at (3,5) {$12$};
\node at (4,5) {$0$};
\node at (5,5) {$1$};
\node at (6,5) {$2$};
\node at (7,5) {$3$};
\node at (8,5) {$4$};
\node at (9,5) {$5$};
\node at (10,5) {$6$};
\node at (11,5) {$7$};
\node at (12,5) {$8$};
\node at (0,4) {$1$};
\node at (1,4) {$2$};
\node at (2,4) {$3$};
\node at (3,4) {$4$};
\node at (4,4) {$5$};
\node at (5,4) {$6$};
\node at (6,4) {$7$};
\node at (7,4) {$8$};
\node at (8,4) {$9$};
\node at (9,4) {$10$};
\node at (10,4) {$11$};
\node at (11,4) {$12$};
\node at (12,4) {$0$};
\node at (0,3) {$6$};
\node at (1,3) {$7$};
\node at (2,3) {$8$};
\node at (3,3) {$9$};
\node at (4,3) {$10$};
\node at (5,3) {$11$};
\node at (6,3) {$12$};
\node at (7,3) {$0$};
\node at (8,3) {$1$};
\node at (9,3) {$2$};
\node at (10,3) {$3$};
\node at (11,3) {$4$};
\node at (12,3) {$5$};
\node at (0,2) {$11$};
\node at (1,2) {$12$};
\node at (2,2) {$0$};
\node at (3,2) {$1$};
\node at (4,2) {$2$};
\node at (5,2) {$3$};
\node at (6,2) {$4$};
\node at (7,2) {$5$};
\node at (8,2) {$6$};
\node at (9,2) {$7$};
\node at (10,2) {$8$};
\node at (11,2) {$9$};
\node at (12,2) {$10$};
\node at (0,1) {$3$};
\node at (1,1) {$4$};
\node at (2,1) {$5$};
\node at (3,1) {$6$};
\node at (4,1) {$7$};
\node at (5,1) {$8$};
\node at (6,1) {$9$};
\node at (7,1) {$10$};
\node at (8,1) {$11$};
\node at (9,1) {$12$};
\node at (10,1) {$0$};
\node at (11,1) {$1$};
\node at (12,1) {$2$};
\node at (0,0) {$8$};
\node at (1,0) {$9$};
\node at (2,0) {$10$};
\node at (3,0) {$11$};
\node at (4,0) {$12$};
\node at (5,0) {$0$};
\node at (6,0) {$1$};
\node at (7,0) {$2$};
\node at (8,0) {$3$};
\node at (9,0) {$4$};
\node at (10,0) {$5$};
\node at (11,0) {$6$};
\node at (12,0) {$7$};
\end{tikzpicture}
\end{center}
\caption{A pattern for coloring the $4^{\text{th}}$ power of the square grid.}
\label{motifdistance4}
\end{figure}

\subsection{Coloring the $k^{\text{th}}$  power of triangular  grids}
The chromatic number of the $k^{\text{th}}$  power of the triangular grid has been determined by Sevcikova~\cite{SE2001}.

\begin{theo}[\cite{SE2001}]
For any $k\ge 1$, $\chi({\mathscr{T}}^{k})= \lceil 3(k+1)^2/4 \rceil$.
\end{theo}
Let $m'_{k}=\lceil 3(k+1)^2/4 \rceil$. We recall the definition of the following function:
$$ f^k_{\mathscr{T}}(i,j)= \left\{
    \begin{array}{ll}
        (i_{\pmod{3(k+1)/2}}+j(3(k+1)/2)+ \\
        \lfloor 2 j/(k+1) \rfloor (k+1)/2)) \pmod{m'_{k}}  & \mbox{if } k \mbox{ is odd}; \\
        (i+(3k/2+1) j) \pmod{m'_{k}}  & \mbox{otherwise.} \\
    \end{array}
\right.
$$
Note that $f^k_{\mathscr{T}}(i,j)=f^k_{\mathscr{T}}(i',j')$, in the case $i\equiv i'\pmod{m'_{k}}$ and  $j\equiv j'\pmod{m'_{k}}$.
This function is used in order to assign $k$-local identifiers to particles.

\subsection{Coloring the $k^{\text{th}}$  power of king grid}
To our knowledge, the chromatic number of the king grid has not been determined yet. However, in contrast with the triangular grid, the chromatic number of the $k^{\text{th}}$ power of the king grid is easy to determine. In this subsection, we determine the exact value of the chromatic number of the $k^{\text{th}}$  power of the king grid.

\begin{theo}
We have $\chi({\mathscr{K}}^{k})= (k+1)^2$.
\end{theo}
\begin{proof}
Let $\mathscr{K}_{k}$ be the subgraph of $\mathscr{K}$ induced by the vertices $\{(i,j)\in V(\mathscr{K})|\ 0\le i\le k,\ 0\le j \le k\}$. 
Note that $diam(\mathscr{K}_{k})=k$ and that $|V(\mathscr{K}_{k})|=(k+1)^{2}$. Thus, since each vertex of $\mathscr{K}_{k}$ must be colored differently in a coloring of the $k^{\text{th}}$  power of the king grid, we obtain that $\chi({\mathscr{K}}^{k})\ge (k+1)^2$.
We define the coloring function $c((i,j))= i_{\pmod{k+1}}+(k+1)j_{\pmod{k+1}}$. Note that we have $d_{\mathscr{K}}(u,v)\ge k+1$, for every two vertices $u$ and $v$ with the same color in $\mathscr{K}$.
Therefore,  we obtain that $\chi({\mathscr{K}}^{k})= (k+1)^2$.
$\hfill\qed$
\end{proof}
Note that since there is a pattern, a vertex $(i,j)$ can determine its color only knowing $i\pmod{(k+1)}$ and $j\pmod{(k+1) }$.
We recall the definition of the following function $f^k_{\mathscr{K}}(i,j)= i_{\pmod{k+1}} +(k+1)j_{\pmod{k+1}}$. Note that $f^k_{\mathscr{K}}(i,j)=f^k_{\mathscr{K}}(i',j')$, in the case $i\equiv i'\pmod{k+1}$ and  $j\equiv j'\pmod{k+1}$.
This function is used in order to assign $k$-local identifiers to particles.
\end{subappendices}
\end{document}